\documentclass[10pt,a4paper]{article}

\usepackage{hyperref}
\usepackage{titling}
\newcommand{\subtitle}[1]{%
  \posttitle{%
    \par\end{center}
    \begin{center}\large#1\end{center}
    \vskip0.5em}%
}

\RequirePackage[UKenglish]{isodate}

\RequirePackage{xcolor}
\RequirePackage{textcomp}
\RequirePackage{amsmath, amssymb, amsthm, stmaryrd}
\RequirePackage[utf8]{inputenc}
\RequirePackage[T1]{fontenc}
\RequirePackage[english]{babel}
\RequirePackage{titling,geometry}
\author{Tomasz Drab}
\title{Reduction Strategies in the Lambda Calculus\\ and Their Implementation through\\ Derivable Abstract Machines: Introduction\thanks{This is the introduction to my Ph.~D. dissertation supervised by Dariusz Biernacki and Małgorzata Biernacka, defended on 19th of June 2023 at the Faculty of Mathematics and Computer Science of the University of Wrocław.}}

\subtitle{(Strategie redukcji w rachunku lambda\\ i ich implementacja za pomocą\\ wyprowadzalnych maszyn abstrakcyjnych: Wprowadzenie)}
\date           {23rd September 2022}

\usepackage[lighttt]{lmodern}
\usepackage{enumitem, ebproof}
\usepackage{listings}
\usepackage{pdfpages}
\usepackage[language=english,style=alphabetic,backend=bibtex,defernumbers=true]{biblatex}
\usepackage[style=british]{csquotes}
\usepackage{cancel, soul}
\definecolor{HlYellow}{rgb}{0.988, 0.906, 0.275}
\sethlcolor{HlYellow}
\renewcommand{\emptyset}{\varnothing}

\newcommand{\nil}{[\,]}
\newcommand{\cons}[2]{{#1 :: #2}}

\newcommand \argfill{\,\cdot\,}
\newcommand \alt{\;\;|\;\;}
\newcommand{\update}[3]{{{#1}[{#2} \!:=\! {#3}]}}
\newcommand{\subst}[3]{\update{#3}{#1}{#2}}
\newcommand{\fv}{ {\mathit{FV}} }
\newcommand{\bv}{ {\mathit{BV}} }
\newcommand{\contr}[1]{ {\rightharpoonup}_{#1} }

\newcommand{\const}[1]{ {\lceil {#1} \rceil} }
\newcommand{\tlam}[2]{{\lambda{#1}. {#2}}}
\newcommand{\tapp}[2]{{{#1}\,{#2}}}

\newcommand{\hole}{\Box}
\newcommand{\plug}[2]{ {{#1}[{#2}]} }
\newcommand{\red}[2]{ {\stackrel{#1}{\to}_{#2}} }
\newcommand{\rred}[2]{ {\stackrel{#1}{\twoheadrightarrow}_{#2}} }
\newcommand{\nred}[2]{ {\stackrel{#1}{\not\to}_{#2}} }
\newcommand{\converts}[2]{ {\stackrel{#1}{=}_{#2}} }
\newcommand{\nv}{\text{naïve}}

\newcommand{\leftstra}{{\scriptstyle\swarrow}}
\newcommand{\rightstra}{{\scriptstyle\searrow}}

\newcommand{\cut}[1]{ \text{\hl{${\big\langle{#1}\big\rangle}$}} }
\newcommand \etrian{{\textcolor{blue}\triangledown}}
\newcommand \ctrian{{\textcolor{green!50!black}\vartriangle}}
\newcommand{\conf}[2]{\langle{#1}, {#2}\rangle}
\newcommand{\econf}[2]{\langle{#1}, {#2}\rangle_\etrian}
\newcommand{\cconf}[2]{\langle{#1}, {#2}\rangle_\ctrian}
\newcommand{\nconf}[2]{\langle{#1}, {#2}\rangle_{\ctrian'}}
\newcommand{\es}[2]{[{#1} \!\leftarrow\! {#2}]}

\newtheorem{lemma}{Lemma}
\newtheorem{theorem}{Theorem}
\newenvironment{dedication}
  {
   \thispagestyle{empty}
   \vspace*{\stretch{1}}
   \itshape             
   \raggedleft          
  }
  {\par 
   \vspace{\stretch{3}} 
   \clearpage           
  }

\begin{document}
\maketitle

\begin{abstract}
\noindent The lambda calculus since more than half a century is a model and foundation of functional programming languages. However, lambda expressions can be evaluated with different reduction strategies and thus, there is no fixed cost model nor one canonical implementation for all applications of the lambda calculus.

This article is an introduction to a dissertation is composed of four conference papers where: we present a~systematic survey of reduction strategies of the lambda calculus;
we take advantage of the functional correspondence as a tool for studying implementations of the lambda calculus by deriving an abstract machine for a precisely identified strong call-by-value reduction strategy;
we improve it to obtain an efficient abstract machine for strong call by value and provide a time complexity analysis for the new machine with the use of a~potential function;
and we present the first provably efficient abstract machine for strong call by need.

\begin{center}\rule[3pt]{300pt}{1pt}\end{center}

Rachunek lambda od ponad pół wieku stanowi model i fundament funkcyjnych języków programowania. Jednak lambda wyrażenia mogą być wartościowane przez różne strategie redukcji i stąd nie ma ustalonego modelu kosztów ani kanonicznej implementacji dla wszystkich zastosowań rachunku lambda.

Niniejszy artykuł jest wprowadzeniem do rozprawy składającej się z czterech prac konferencyjnych, w których: prezentujemy systematyczny przegląd strategii redukcji w rachunku lambda;
wykorzystujemy odpowiedniość funkcyjną jako narzędzie do studiowania implementacji rachunku lambda przez wyprowadzenie maszyny abstrakcyjnej dla precyzyjnie zidentyfikowanej strategii silnego wołania przez wartość;
ulepszamy ją, by otrzymać wydajną maszynę abstrakcyjną dla tej samej strategii, i~dostarczamy analizę złożoności czasowej dla nowej maszyny z zastosowaniem funkcji potencjału;
prezentujemy pierwszą dowodliwie wydajną maszynę abstrakcyjną dla silnego wołania przez potrzebę.
\end{abstract}

\newpage
\begin{dedication}To All Saints\end{dedication}

\setlength{\parskip}{0.8ex}
\section*{Acknowledgements}
\label{sec:ack}
\addcontentsline{toc}{section}{\nameref{sec:ack}}

First, I would like to thank three persons that have the most direct impact on the successful preparation of this work, to my mentors that in the way turned themselves into my co-authors:

To Dariusz Biernacki who plays the role of my primary supervisor,
where by ``to play'' I do not mean ``to pretend'' but ``to fulfil'' the duty authentically and with care.
During our talks, I was able to investigate my doubts and critical thoughts about the essence of research, but it was also an occasion to watch his inspiring ethos of scholar.
However random my questions were, thanks to his guidance it all lead to the completion of the dissertation.

To Małgorzata Biernacka, the auxiliary supervisor of this dissertation and the supervisor of my master thesis.
She is the person that made me realize the non-technical part of the research most.
For example, she showed me the importance of the narrative in research papers and supplied it for our results.
Many times she edited and corrected my messy writing, and I still learn from it how to do it better.
Both my supervisors found a topic of study for me and brought me to an edge of the research where I could try to simply employ my technical skills with an interest in the subject.
It seems to be another crucial element for the completion of my Ph.D. programme.

To Witold Charatonik who informally can be named my third supervisor.
We have discussed the most technical parts of our publications and thanks to him many lemmas gained their formal character, in particular proofs of the bypass transitions' correctness.
Many times, I could witness his true conversance and proficiency in formal reasoning.
At the same time, he is for me a role model of humility proper to scholars.

I also would like to note how lucky I was to have had great teachers continuously over the years and to thank just a few of them. 
To Sławomir Kosiński who blessed me with a vision of a future doctorate in mathematics when I was in primary school and to Alina Krzykawiak who at the time taught me to write my first programs in Logo.
To Stanisława Grzywna who oversaw my mathematical training in junior high school.
To Agnieszka Kazun, Beata Laszkiewicz, and Michał Śliwiński, my high school teachers, who taught me mathematics and informatics in such an inspirational way that I could get to know them as persons at the same time.
And to Yuriy Kryakin, the supervisor of my bachelor thesis, who spent many hours and much of his attention sharing with me his affection for mathematics, resulting in renewing my own.

I am grateful to the Foundation of Mathematicians of Wrocław for organizing mathematical competitions for children that motivated me to engage more with mathematics.
Especially the contest called KoMa affected my approach to research and learning.
I~would like to draw attention also to the scientific environment organized at the Institute of Computer Science of the University of Wrocław.
The work of its employees and students enabled me to obtain higher education in a great atmosphere.
Most abstractly, I am grateful for an opportunity to work as just a teacher and researcher thanks to public funding.

My research perspective gained much from the contact with researchers representing the external research community.
I would like to thank Beniamino Accattoli for remote discussions on technical topics and Claudio Sacerdoti Coen for hosting a fruitful meeting of us three at the University of Bologna and presenting the charming city of Bologna to me. 

I thank my friends that help me in so many various ways that it would be too hard to enumerate it all here.

Finally, I want to thank my family for numerous signs of their support.
Especially, to my parents, Danuta and Piotr, for their, expressed repeatedly for almost three decades of my life,
unconditional care for my good
and belief that I am able to participate in great things.
And to my wife, Monika, for inspiration and her one-of-a-kind help during almost seven years of marriage.
\setlength{\parskip}{1.5ex}

\tableofcontents

\section{Social Introduction}
\subsection{Prologue}

Many have undertaken the effort of setting down an account of their research experience in the form of a dissertation. After having been investigating the topic from the ground up, it falls on me too to write an orderly account of the phenomena that I have encountered during the last four years of my work. I am willing to do it, so that everyone may follow the way we went and learn about these things easier.

However, I try to approach the task humbly \cite{DBLP:journals/cacm/Dijkstra72}, being aware of my finitude and ignorance. Because of them, my report will be incomplete and possibly flawed in places. Regardless, I discern that the best I can do is to, still uneducated and biassed as I am, and in limited time, honestly describe the things as they appear to me.

\subsection{Outline of the Dissertation}

The dissertation begins with the general introduction divided into three subchapters:
\begin{itemize}
\item \textit{Social Introduction}
\item \textit{Methodological Remarks}
\item \textit{Technical Introduction}
\end{itemize}
and a short \textit{Conclusion} section. The first two subchapters present freely my personal motivations and reflections upon research work while the third one is intended to introduce the reader smoothly to the technical content of the dissertation keeping both feet on the ground and discuss its contributions.

The introduction is followed by four technical chapters that are research papers published at various conferences:
\begin{itemize}
\item[{\cite{I}}] Małgorzata Biernacka, Witold Charatonik, Tomasz Drab. \textit{The Zoo of Lambda-Calculus Reduction Strategies, And Coq}. ITP 2022: 7:1-7:19
\item[{\cite{II}}] Małgorzata Biernacka, Dariusz Biernacki, Witold Charatonik, Tomasz Drab. \textit{An Abstract Machine for Strong Call by Value}. APLAS 2020: 147-166
\item[{\cite{III}}] Małgorzata Biernacka, Witold Charatonik, Tomasz Drab. \textit{A Derived Reasonable Abstract Machine for Strong Call by Value}. PPDP 2021: 6:1-6:14
\item[{\cite{IV}}] Małgorzata Biernacka, Witold Charatonik, Tomasz Drab. \textit{A Simple and Efficient Implementation of Strong Call by Need by an Abstract Machine}. ICFP 2022: 94:1-94:28
\end{itemize}

\subsection{On the Effectiveness of Language}

The root of my research endeavour is a reflection on the essence of language, mathematics, and science in general. In the following paragraphs, I will refer jointly to logic and mathematics, given that the lion's share of mathematics can be logicized. I will also to some extent conflate them with language because language, more and more formal, is our gateway to logic and mathematics.
Since this philosophical reflection is not the topic of the dissertation itself, the selection of literature referenced below is not exhaustive, but rather illustrative.

A common practice to introduce such context is to refer to Eugene Wigner’s article from 1960 titled “The Unreasonable Effectiveness of Mathematics in the Natural Sciences” \cite{Wigner60}. Wigner in the beginning mentions his “eerie feeling” about unexpected applications of abstract mathematics. Thereafter, he states that mathematical language is “the correct language” to formulate laws of nature.

Even earlier, in 1936, Albert Einstein expressed his awe over “the very fact that the totality of our sense experiences is such that by means of thinking [\ldots] it can be put in order” in words “the eternal mystery of the world is its comprehensibility” and calls it a miracle \cite{Einstein36}. Again, I find the linguistic aspect of the phenomenon in the abstract of Einstein’s article where he states “physics constitutes a logical system of thought.”

More recently, Edward Frenkel in the book „Love and Math: The Heart of Hidden Reality” shared his delight over the beauty and power of mathematics. Frenkel witnesses mathematics’ objectivity and constancy, that is the meaning of mathematical theorems is independent of time, space, gender, religion, or skin colour \cite{Frenkel13}. Moreover, applications of formal languages are still reported. For instance, scholars from various universities elaborated on examples of the “Unusual Effectiveness of Logic in Computer Science” not long ago \cite{DBLP:journals/bsl/HalpernHIKVV01}. Plus, from a different angle, we value the truth in our day-to-day lives. Perhaps needless to say, the tool to state the truth is language.

As an explanation of the phenomenon, Max Tegmark proposes and defends a hypothesis that the universe itself is a mathematical structure \cite{Tegmark14}.
This perspective may be attractive to some with a reductionist proclivity akin to mine.
After watching the film “The Matrix”, it challenges us less to imagine that we perceive and live in a computer simulation.

In any case, it is plausible to me that the universe is subjected to an exhaustive description in a formal language. It fascinates me that reality, or at least its parts, is simulable in computer or human memory, so that knowledge, in the form of the description, can be extrapolated to reliable predictions. I consider it a wonder that something as deep as meaning can be captured and reified into something as flat, material and mechanically processable as text. I also find it useful to articulate and verbalize our thoughts in order to objectively, possibly externally, analyze their validity because the rules of logic are already set and insusceptible to manipulation. It is how thinking is elevated to reasoning. In extreme cases, the analysis may take hundreds of years, as it was in the case of Fermat’s Last Theorem. All of this brings me to a contemplation of concepts such as language, meaning, word, truth, reason, dialogue etc. My research activity was a modest attempt to commune with the \emph{Phenomenon of the effectiveness of language} described in the preceding paragraphs.

Besides witnesses of the effectiveness of language, logic and mathematics, there are critics of such an enthusiastic view who give refreshing arguments for alternative explanations. For example, we shall consider if the presented perspective is not fragmentary due to an adaptation formed in the process of Darwinian evolution. Nevertheless, I think that stating my viewpoint, mimicking the scientific debate, and dialoguing are appropriate means to get rid of my delusions, including these about the Phenomenon.

Studies of formal languages resulted in conclusions of the limitative character with Gödel's incompleteness theorem, as conceivably the most widely known, at the forefront. Another incompleteness of language is that expressions of a given language are always subject to interpretation. Even a formally defined language still remains dependent on the interpretation of the metalanguage. However, it is not obvious how to circumvent the limitations of language.

Existence of natural order and the human ability to discover it are both underlined by Mariano Artigas as presuppositions of scientific endeavour \cite{Artigas00}. What is more, alongside the “intelligibility”, as the third presupposition of science Artigas enumerated an assertion that the search for truth is valuable, in other words, worthwhile. Under this assumption, many scholars, and I too, could devote their professional work to fundamental research such as the study of the lambda calculus.

\subsection{On a Universal Language}%
Henk Barendregt and Erik Barendsen in their “Introduction to Lambda Calculus” mention Gottfried Leibniz’s dream to create a universal language of problem description and a procedure to solve problems expressed in it \cite{BarendregtB84}. I maintain that, in a sense, as far as possible, Leibniz’s dream came true.

In the 1930s, the notion of algorithm, that is, strictly procedural routine, was formalized and accepted thanks to the conception of Turing machine \cite{Turing37}, and lambda calculus \cite{Church32} turned out to be a language that can express any computation that the machine can execute.

However, it is not that lambda calculus is just a language for programming in an arbitrary mathematical model. There were many models of “effective calculability” or “mechanical method” proposed. Some of them are Kurt Gödel's general recursive functions, Alonzo Church's lambda calculus, Alan Turing's machines, Andrey Markov Jr.'s string rewriting, and Andrey Kolmogorov's abstract notion of algorithm. All of them turned out to identically specify what is computable and what is not. For example, values of the Ackermann-Péter \cite{Peter35} function can be computed in any of them, while there is no general algorithm comparing two real numbers in none of them. It leads to a conviction, expressed in Church-Turing thesis, that any reasonable model of computation is equivalent to these enumerated above in the same sense as they are equivalent to each other \cite{Church36}. Thus computability seems to be a robustly defined notion, resilient to change in technical details as it would exist independently of human interest.

The construct of lambda abstraction was already present in John McCarthy’s Lisp programming language \cite{DBLP:journals/cacm/McCarthy60}. However, the fact that expressions of the lambda calculus themselves can be used as a programming language was spelled out by Peter Landin in the 1960s \cite{Landin65a,Landin65b}, and the calculus became the foundation of functional programming languages. Landin also offered a method of “mechanical evaluation” of lambda expressions via the SECD abstract machine designed for this purpose \cite{Landin64}. Despite the minimalism of the calculus, it enables use of higher-order functions by default, and thus programming at a high level of abstraction.

In summary, lambda calculus constitutes a very simple, functional, universal programming language. This makes the calculus an interesting research topic as it touches the essence of computability. In turn, abstract machines for the lambda calculus make visible the finitude behind the notion of algorithm that epitomizes human's finitude, and that is apparent in the head of Turing machine, single Markov's rule, or stated explicitly in Kolmogorov's definition as ``bounded complexity'' \cite{kolmogorov53}.
\newpage

\subsection{Practical Applications}
\label{sec:applications}
A belief in the profitableness of the lambda calculus research could be naive if it was not supported by its applications and was based only on the universality of the calculus.
The reason is that there are many Turing tarpits
\cite{DBLP:journals/sigplan/Perlis82} that are Turing-complete systems but are not fit to be practically used.
Although an enumeration of practical applications may be mundane,
it is nonetheless important.
Fortunately, the hopes placed in the lambda calculus are being fulfilled.

Just to illustrate, the lambda calculus is a subset of general-purpose, functional programming languages such as OCaml, Racket, and Haskell.
Thus, implementation issues concerning the calculus translate into matters of these languages.
Nowadays, the construct of lambda expression is present even in languages of great popularity and with wide applications such as Python, Java, and C++.


The lambda calculus was intended as the central object of study of this dissertation.
However, the area of lambda calculus research is still very wide.
Therefore, this research was focused on the operational semantics of the pure lambda calculus with particular emphasis on strong reduction.
It was chosen as an entry point to the study of lambda calculus because of small dependence on other theoretical constructs (Chapter~I presents fairly self-contained Coq formalization of more than twenty small-step semantics) and potential reuses in close research lines.
The research concerned also with an efficient implementation of small-step semantics in order to enable practical experimenting with them.
Maybe more importantly, efficient implementation of strong reduction finds applications in proof assistants, partial evaluation, and compilation of modules as described below.


Type checking in proof assistants requires to decide if given two types are equal, that is, more precisely, convertible.
In general, the $\beta$-convertibility of two lambda terms is undecidable because otherwise, we could decide the extensional equality of functions, and the halting problem.
However, in proof assistants, the type system ensures that the terms are normalizable, so both given terms can be normalized and then easily compared.

The goal of partial evaluation can be understood as a specialization of a function of two parameters, static and dynamic, given the static one.
Optimizations of the static parts take place under the abstraction of the dynamic parameter and thus can be seen as an example of strong reduction.

Most recently, Gabriel Scherer and Nathanaëlle Courant reported in the form of workshop presentation \cite{SchererC22} a use case of strong reduction in the compilation of OCaml's module system.
OCaml functors play there role of lambda abstractions.
Strong normalization of functors is executed in the service of the programmer during the compilation in order to locate places where values are declared and defined.

The ultimate test of the usefulness of the research presented in this dissertation will be the degree to which its results will be used.
However, the results were published very recently so it cannot be known at the moment.
Nevertheless, for example, the implementation of strong call by need in the OCaml's modules normalizer resembles much the normalizer of Chapter IV, so the upper bound obtained from its complexity analysis may be already applicable to the part of the OCaml compiler.


\newpage

\section{Methodological Remarks}

It was not until $11^\text{th}$ of October 2019, during the third semester of my doctoral studies, when, influenced by what I had been taught, I~implemented my own compositional full-reducing evaluator for the lambda calculus. Soon it turned out that it is connected with Crégut’s KN machine \cite{DBLP:journals/lisp/Cregut07}, indirectly through the functional correspondence \cite{DBLP:conf/ppdp/AgerBDM03}. Later I learned that I had reproduced a result from Johan Munk’s master thesis \cite{Munk:MS}. In January next year, I completed a derivation of the call-by-value variant of KN that became the titular contribution of Chapter~II. However, it was not until July, when we submitted the paper for a conference. Therefore, this subchapter is devoted to the reflection upon my work as a researcher that turned out to not come down to making alienated discoveries.

\subsection{Exploration of the Graph of Potential Knowledge}

For a few years now I have been one of those who imagine the body of all knowledge structured as a graph. Using the language of Platonic perspective, the graph I think of also contains what we do not know yet, that is the potential knowledge.

A helpful imagination of such a graph is a set of hypertext documents connected with hyperlinks between them. Wikipedias and Wikidata are good examples of such a structure. Furthermore, these projects, if properly read, constitute a useful approximation of the part of the Graph of Potential Knowledge which is already commonly known.

I conceptualize a part of a researcher's work as an exploration of this graph on the border of what we know and what we would like to know. I imagine that the edges of the graph have diverse weights because theorems are accessible, both by formal proof and by intuition, with different difficulties. Therefore, I think I follow the approach of Dijkstra's algorithm \cite{DBLP:journals/nm/Dijkstra59} on the subgraph of my research interests. I try to explore its neighbourhood to find easy paths of reasoning about problems considered hard. It is in order to follow the KISS (keep it simple) principle and not fall into a complexity that I would not be able to comprehend \cite{DBLP:journals/cacm/Dijkstra72}. It is where a metacognitive perspective helps to impose some discipline of thinking. In my experience, choosing the easiest questions from the pool of questions that have to be investigated anyway decreases the probability of getting stuck, and thanks to small successes, it increases the motivation to focus on research work. Nevertheless, I still sometimes get stuck because I forget about this guideline.

This approach brought me to a formalization of various reduction strategies and their properties that became the base of a research paper and Chapter I. Its motivation was also to pay off a bit of the research debt \cite{olah2017research} which is an analogue of technical debt in research.
It may arise when a hard research aim is accomplished,
but the presentation of work is not refined, measures taken were more difficult than needed, knowledge distributed in various research papers is not organized in one source etc.
It makes it harder to reuse results and techniques used by other researchers.
The problem may be summarized in words of William Thurston: ``mathematicians usually have fewer and poorer figures in their papers and books than in their heads'' \cite{Thurston94}. 
The research debt may for example slow down the pace of new researchers entering the area.

In my case, I have been experiencing a continued confusion about simple question if the innermost-head-spine strategy that appeared in Peter Sestoft’s paper \cite{DBLP:conf/birthday/Sestoft02} twenty years ago is a separate strategy from well-known head reduction in the sense of small-step operational semantics or not. Plausibly, this was obvious to experts, but not elaborated in the literature. I hope that now anyone with similar doubts will likely find our open-access publication that decides this one and similar questions with Coq proofs.

At least to the extent researcher’s work in the field of computing is based on formalizable reasoning, it resembles programmer’s work via Curry-Howard correspondence. Another rule of major importance to programmers is the DRY (don’t repeat yourself) principle. I think that one of the main motivations of programming is to complete repetitive tasks once and for all. Thus, repeating the same code would go against one of the purposes of programming. Less code means also less writing, reading, maintenance etc., and in general less work and less opportunities to make a mistake. Therefore, I try to follow the DRY principle also in research by developing reusable resources.
It involves contributing to Wikimedia projects, disclosing code developments, and sharing other supplementary materials.

While searching the Graph of Potential Knowledge, we can encounter issues that appear to be new, but are only unusual representations of what is already known. In order to decrease the number of concepts that I have in my head, I~try to name the things I come across and equalize them with everything equivalent early. I do it additionally also in my native language to capture the essence of a~given concept even better, often from a different perspective.

There is little literature in Polish on this topic. Therefore, I attach below a~list of translations that I currently use for this purpose for interested readers.

\subsection{Used Polish Translations}

\begin{tabular}{p{6cm}l}
abstract syntax            & składnia abstrakcyjna    \\
big-step semantics         & semantyka dużych kroków  \\
binder                     & wiązar                   \\
bypass                     & obejście                 \\
call by name               & wołanie przez nazwę      \\
call by need               & wołanie przez potrzebę   \\
call by value              & wołanie przez wartość    \\
closure                    & domknięcie               \\ 
conclusion                 & wniosek                  \\
concrete syntax            & składnia konkretna       \\
conservative extension     & zachowawcze rozszerzenie \\
context                    & kontekst                 \\
contrex                    & kontreks                 \\
continuation               & kontynuacja              \\
continuation-passing style & styl kontynuacyjny       \\
definitional interpreter   & interpreter definiujący  \\
decomposition              & rozkład                  \\
direct style               & styl bezpośredni         \\
dump                       & zrzut                    \\
eager                      & gorliwy                  \\
functional correspondence  & odpowiedniość funkcyjna  \\
ghost machine              & maszyna widmo            \\
if … then … else …         & jeśli … to … inaczej …   \\
implicit sharing           & niejawne współdzielenie  \\
inert term                 & term bezwładny           \\
instruction                & rozkaz                   \\
lambda calculus            & rachunek lambda                 
\end{tabular}

\begin{tabular}{p{6cm}l}
lazy                       & leniwy                   \\
memoization                & spamiętywanie            \\
normal order               & porządek normalny        \\
optional                   & opcjonał                 \\
partial evaluation         & częściowa ewaluacja      \\
pattern matching           & dopasowanie wzorca       \\
potential function         & funkcja potencjału       \\
premiss                    & przesłanka               \\
redex                      & redeks                   \\
reduction semantics        & semantyka redukcyjna     \\
refocusing                 & przeogniskowanie         \\
renaming                   & przemianowanie           \\
rigid term                 & term sztywny             \\
shape invariant            & niezmiennik kształtu     \\
small-step semantics       & semantyka małych kroków  \\
stack                      & stos                     \\
statement                  & instrukcja               \\
strong reduction           & silna redukcja           \\
structural operational
   semantics               & strukturalna semantyka
                                operacyjna            \\
stuck term                 & term zacięty             \\
uniform strategy           & strategia jednostajna    \\
weak reduction             & słaba redukcja          
\end{tabular}

\subsection{Empirical Method}

In the presence of Gödel's incompleteness theorem, Gregory Chaitin encourages mathematicians to bend to the empirical method because some theorems may not be provable in the classical sense.

An example of empirical evidence is given in Section 4.2 of Chapter IV. Abstract machines allow us to measure an exact, integer number of machine steps for a given input. This way we can map a sequence of inputs into a numerical sequence. For the initial terms of the sequences reported there, a closed formula has been found and presented. Because of the moderate complexity of inputs, it seems that the formula is correct for all subsequent terms. However, since the claims concerning closed formulas of complexities of particular term sequences are relatively less important, we spared an effort of preparing their formal proofs. Nevertheless, the statements are clear and most importantly falsifiable. A single counterexample would suffice to refute any of them.

Experimentation with implemented abstract machines can also help to study their properties. Before the RKNV machine of Chapter III was derived, I derived an improvement of KNV of Chapter~II that memoized neutral terms but not abstractions. In such a setting it was nearly effortless to test given sequences of inputs and observe if the exponential growth is present in the output sequence. The counterexample showed when we could not complete a proof of polynomial overhead of the leaky machine. Then the derivation of RKNV had to start over from a new normalizer. However, having the RKNV we could reuse the tests prepared for the previous machine. Then thanks to watching the machine runs generated for these examples, it was easier to complete the desired proof.
Currently, the cost of a new derivation is significantly reduced as it has been automated by Maciej Buszka and Dariusz Biernacki \cite{BuszkaB21:LOPSTR21}. We have benefited from it while constructing RKNL, and we describe this particular derivation in Chapter IV.

In my opinion, these examples show how computer experiments may accelerate research thanks to instant responses from it and thus that interactive theorem proving may be supported by both proof assistants and software development environments.

\subsection{Science Communication}

As I mentioned at the beginning of this subchapter, a researcher's work does not come down to making discoveries. What use would be of discoveries if they would not be shared with others? Therefore, a researcher also has to communicate the results of their work. Thus, to my surprise, it turned out that even a computer scientist has to use not only technical skills to operate on abstract and relatively simple objects but also soft skills to interact with concrete subjects with all their complexity because at the end of the day we are all people.

Traditionally, researchers articulate their thoughts in writing, which itself is a challenging task. Taking into account many aspects of human psychology may be very beneficial to describe a new idea. It includes arousing the interest, clarifying the presentation, setting the narrative, limiting the delay of gratification etc.

A major aspect of the scientific method is reproducibility. To ease replication of results, every research paper that makes up a chapter of this dissertation is accompanied by code. I believe that formal languages are also effective means of communication between people.

The development of the Internet facilitates novel forms of communication to take advantage of. My academic website\footnote{\href{https://ii.uni.wroc.pl/~tdr/}{https://ii.uni.wroc.pl/\~{}tdr/}} is planned to present an organized view of my scientific activity, especially this part that will not be covered by this dissertation. Additionally, most of the conference presentations of the four technical chapters will be available as online videos linked by the website. More importantly, my website should provide my email address so, where possible, I can answer questions and present my explanation of matters that may be unclear.

Finally, I would like to briefly mention the role of the word between the two, that is of the dialogue. I could experience that, even in the professional setting, it can turn a disagreement and deep-seated, mutual misunderstanding into a nuanced and more accurate perspective.
Therefore, I think it is worthwhile to exercise what is at the service of dialogue.

\newpage

\section{Technical Introduction}

The technical introduction constitutes preliminaries of technical content of the dissertation.
It is intended to contain simple examples that demonstrate concepts considered in the next chapters and to present broader context of their contributions.

\subsection{Metalanguage}

A \emph{metalanguage} is a language used to describe another language, called the object language or the \emph{defined language} in case of a definition.
As an example let us take CPython that is, to simplify, an implementation of Python language written in C.
If we would take CPython as a definition of the Python language, then the C language plays the role of the metalanguage and Python is the defined language.

We are about to use the English language with borrowings from mathematical formalism to define the lambda calculus from the ground up.
However, first we will define a simple language of additive expressions to illustrate the definition process
and to set the stage for simple examples of reduction strategies and abstract machines.
Later we will also use OCaml as a metalanguage to define semantics, that is meaning, of additive expressions again.

\subsection{Additive Expressions}

We define abstract syntax of additive expressions by a context-free grammar:
$$ \mathcal E \ni  e ::= \const n \alt e_1 \oplus e_2 $$
where metavariable $n$ ranges over the set of natural numbers $\mathbb N$. The two types of expressions are constants and sums. There are no variables in the defined language.
The idea is that the expression $\const 2 \oplus \const 2$ models the metaexpression ``$2 + 2$'' of our language in the defined language,
so parts of language can become object of our study.  

We use the equality sign = to relate only objects identical at the metalevel. Thus, we have $\const 2 \oplus \const 2 \neq \const 4$ because from the perspective of the metalevel, $\oplus$ does not denote addition, but is an expression constructor.
Using mathematical notation we can describe $\oplus$ as a function that takes two expressions and returns an~expression as follows: $\oplus : \mathcal E \to \mathcal E \to \mathcal E$.  

To denote metalevel addition, as is customary, we will use the plus sign $+$ with the following type at the metalevel $+ : \mathbb N \to \mathbb N \to \mathbb N$.
Thus, the well-known arithmetic fact $2 + 2 = 4$ still holds in our notation.

We work with abstract syntax, as opposed to concrete syntax which is linear, so we understand expressions as trees of derivations from the grammar.
We use parentheses to disambiguate the derivation of an expression from the grammar.
Thus, we also have $\const 3 \oplus (\const 4 \oplus \const 5) \neq (\const 3 \oplus \const 4) \oplus \const 5$, while $3 + (4 + 5) = (3 + 4) + 5$.
The operation that $\oplus$ models, i.e. addition, is associative, but, as we have seen, as a constructor of additive expressions, $\oplus$ is non-associative.

Having formally defined the syntax of the calculus, practically as a datatype, we can programme, that is formally define, functions operating on expressions as, for example, depth of an expression:
\begin{alignat*}{3}
\text{depth}\phantom{(}&&: \mathcal E \to \mathbb N\\
\text{depth}(&&\const n) &= 0\\
\text{depth}(&&e_1 \oplus e_2) &= 1 + \max(\text{depth}(e_1), \text{depth}(e_2))
\end{alignat*}
Note that the expressions $\const 2 \oplus \const 2$ and $\const 4$ have different depths: $$\text{depth}(\const 2 \oplus \const 2) = 1 \neq 0 = \text{depth}(\const 4).$$
Since they have even different depths, they cannot be equal.
However, we would like to model that the the result of evaluation of expression ``2 + 2'' is ``4''.

For this purpose, we define a \emph{contraction} relation $\contr +$ that simplifies sums of constants in the following way: $\const{n_1} \oplus \const{n_2} \;\contr +\; \const{n_1 + n_2}$.
Then we have $\const 2 \oplus \const 2 \;\contr +\; \const 4$.
However, still $(\const 3 \oplus \const 4) \oplus \const 5 \not\rightharpoonup_+ \const 7 \oplus \const 5$ because the subexpression $\const 3 \oplus \const 4$ is not of the form $\const n$.
As usually, the contraction relation denotes a local operation that needs to be extended in order to operate on the whole expression globally.
Therefore, we want to define a \emph{reduction} relation that can simplify also the sums of subexpressions.

To this end, we define notion of \emph{context} by the following grammar.
$$\mathcal G \ni G ::= \hole \alt G \oplus e \alt e \oplus G$$
A context can be understood as a expression with exactly one \emph{hole} (denoted $\hole$).
A~context can serve to point a position of a subexpression.
For example, a context $\hole \oplus \const 5$ is a context of the subexpression $\const 3 \oplus \const 4$ in the expression $(\const 3 \oplus \const 4) \oplus \const 5$.
The whole expression can be recomposed using the \emph{plug} function defined below.
\begin{alignat*}{3}
\plug \argfill \argfill &: \mathcal G \to \mathcal E \to \mathcal E\\
\plug \hole {e_0}  &= e_0\\
\plug {(G \oplus e)} {e_0} &= \plug G {e_0} \oplus e\\
\plug {(e \oplus G)} {e_0} &= e \oplus \plug G{e_0}
\end{alignat*}
Then we have
$\plug {(\hole \oplus \const 5)} {\const 3 \oplus \const 4} =
(\plug \hole {\const 3 \oplus \const 4}) \oplus \const 5 =
(\const 3 \oplus \const 4) \oplus \const 5$.
Having the plug function, we can extend any contraction $\contr \iota$ to a reduction $\red {} \iota$ by a \emph{contextual closure}.

\begin{center}\begin{prooftree}
\hypo{e_1 \;\contr \iota \; e_2}
\infer1{\plug G {e_1} \;\red{}\iota\; \plug G {e_2}}
\end{prooftree}\end{center}

Finally, we can state that  $(\const 3 \oplus \const 4) \oplus \const 5 \red {} + \const 7 \oplus \const 5$
and $\const 7 \oplus \const 5 \red {} + \const {12}$.
Thus, we can say that $(\const 3 \oplus \const 4) \oplus \const 5$ reduces to $\const {12}$ in exactly two steps and cannot reduce to it in one step $(\const 3 \oplus \const 4) \oplus \const 5 \not \to_+ \const {12}$.

The reflexive-transitive closure of \emph{reduction in one step} $\red {} \iota$ is denoted by $\rred {}\iota$, and called \emph{reduction in zero or more steps}.
The reflexive-symmetric-transitive closure is denoted by
$\converts {}\iota$, and is called \textit{conversion}.
Therefore, $(\const 3 \oplus \const 4) \oplus \const 5$ and $\const {12}$ are convertible to each other: $(\const 3 \oplus \const 4) \oplus \const 5 =_+ \const {12}$.

Note the difference between the metalevel equality = and conversion $=_+$ obtained by the contextual and reflexive-symmetric-transitive closures:
$\const 2 \oplus \const 2$ and $\const 4$ are two different expressions:
$\const 2 \oplus \const 2 \neq \const 4$,
but they are convertible to each other:
$\const 2 \oplus \const 2 =_+ \const 4$.
Similarly, $\oplus$ is not associative with respect to metalevel equality: $\const 3 \oplus (\const 4 \oplus \const 5) \neq (\const 3 \oplus \const 4) \oplus \const 5$, but is associative w.r.t. $+$-conversion: $\const 3 \oplus (\const 4 \oplus \const 5) =_+ (\const 3 \oplus \const 4) \oplus \const 5$.

This way we have defined small-step \emph{reduction semantics} \cite{DBLP:journals/tcs/FelleisenH92} of additive expressions.
The reduction $\red {} +$ defines a single, small step of computation,
and the conversion $=_+$ relates the expression, among others, with the irreducible result of the whole computation.
Convertible expressions, as $\const 2 \oplus \const 2$ and $\const 4$, are considered to have the same meaning.
The meaning is not trivial because there are expressions that are inconvertible $\const 0 \neq_+ \const 1$,
so their replacement would be meaningful.

\subsection{Basics of Lambda Calculus}

The abstract syntax of the pure lambda calculus is given by the following context-free grammar:
$$ \Lambda \ni  t ::= x \alt \tapp{t_1}{t_2} \alt \tlam x t $$
where the metavariable $x$ ranges over some set of identifiers $\mathcal X$. The three types of expressions are variables, applications and abstractions. In a sense, applications correspond to specialization and abstractions to generalization.

Similarly, we use parentheses to disambiguate the derivation.
Moreover, we use unambiguous notational shorthands, including left associativity of the application: $\tapp{\tapp xy}z = \tapp{(\tapp xy)}z \neq \tapp x{(\tapp yz)}$, extension of a body of an abstraction as far as possible: $\tlam x {\tapp xx} = \tlam x {(\tapp xx)} \neq \tapp {(\tlam x x)} x$, and coalescing consecutive abstractions into one: $\tlam {xy} x = \tlam x {\tlam y x}$. Examples of lambda expressions are given in Table~\ref{tab:terms}.

\begin{table}[h!!]
\hspace{2.7cm}
\begin{tabular}{rcll}
$I$             & = & $\tlam x x$         & identity                             \\
$K$             & = & $\tlam {xy} x$        & constant function constructor        \\
$S$             & = & $\tlam {xyz} { \tapp{\tapp xz}{(\tapp yz})}$  & parametrized application constructor \\
{$\omega$}        & = & $\tlam x {\tapp xx}$        & self-applicator                      \\
{$\Omega$}        & = & $\tapp \omega\omega$       & canonical divergent term             \\
{\textit{pair}} & = & $\tlam {xyf} {\tapp {\tapp fx}y}$     & pair constructor                     \\
$c_0$            & = & $\tlam {fx} x$        & Church numeral 0                     \\
$c_1$            & = & $\tlam {fx} {\tapp fx}$       & Church numeral 1                     \\
$c_2$            & = & $\tlam {fx} {\tapp f{(\tapp fx)}}$    & Church numeral 2                     \\
$c_3$            & = & $\tlam {fx} {\tapp f{(\tapp f{(\tapp fx)})}}$ & Church numeral 3                     \\
\end{tabular}
\caption{Example lambda terms}
\label{tab:terms}
\end{table}

As we can see, there is abuse of notation in using $x$ as a metavariable and a~concrete name of object variable
at the same time.
It would be less ambiguous to use a separate metavariable as a nonterminal representing variables. If we would take $\varkappa$ as such metavariable, we would have $\Lambda \ni  t ::= \varkappa \alt \tapp{t_1}{t_2} \alt \tlam \varkappa t$. Then, we could say that $x$, $y$, $z$, $f$ are just elements of $\mathcal X$.
However, the convention with a~double role of $x$ is so widespread that we will stick to it.

Similarly as for additive expressions, we can define functions working on lambda terms. Below, sets of \emph{free} and \emph{bound variables} in a term, and the \emph{capture-agnostic substitution} function are presented:

\begin{alignat*}{7}
\fv &: \Lambda \to \mathcal P_{\text{fin}}(\mathcal X) &
\bv &: \Lambda \to \mathcal P_{\text{fin}}(\mathcal X)\\
\fv{(x)} &= \{ x \} &
\bv{(x)} &= \emptyset\\
\fv{(\tapp {t_1} {t_2})} &= \fv{(t_1)} \cup \fv{(t_2)} &
\hspace{12mm}
\bv{(\tapp {t_1} {t_2})} &= \bv{(t_1)} \cup \bv{(t_2)}\\
\fv{(\tlam x {t})} &= \fv{(t)} \setminus \{ x \} &
\bv{(\tlam x {t})} &= \bv{(t)} \cup \{ x \}
\end{alignat*}
\begin{alignat*}{1}
\subst \argfill \argfill \argfill &: \Lambda \to \mathcal X \to \Lambda \to \Lambda\\ 
\subst {x} {t} {{x'}} &= \begin{cases} t &: x = x' \\ {{x'}} &: x \neq x'\end{cases}\\
\subst x {t} {(\tapp {t_1} {t_2}) } &= \tapp {\subst x {t} {t_1}} {\subst x {t} {t_2}}\\
\subst {x} {t} {(\tlam {x'} {t'})} &= \begin{cases} \tlam {x'} {t'} &: x = x' \\ \tlam {x'} {\subst {x} {t} {t'}} &: x \neq x'\end{cases}
\end{alignat*}

The terms with no free variables are called \emph{closed}, and the others are called \emph{open}. A context in the lambda calculus can be seen as a term with exactly one free occurrence of a special variable $\hole$. The notation $\plug C t$ is then a shortcut for $\subst \hole t C$. Assuming that $\hole$ is not used as a bound variable, contexts in the lambda calculus can be defined by the following grammar:
$$C ::=
  \tapp C t \alt
  \tapp t C \alt
  \tlam x C \alt
  \hole$$

The name lambda calculus comes from the letter lambda of abstraction syntax that serves to define functions. Since the choice of the letter lambda is accidental, so is the name lambda calculus. To get closer to its essence, it can be thought of as a calculus of functions or a calculus of abstractions.
The intended meaning of a~lambda abstraction is as follows (on an example of $\tlam x {\tapp xy}$):
\begin{center}
\begin{tabular}{|l|c|c|c|c|}
\hline
notation: & $\lambda$                     & $x$  & .           & $\tapp x y$              \\
\hline
reading:  & a function that takes & ex & and returns & ex applied to wye \\
\hline
\end{tabular}
\end{center}
 
Thus it would be intuitive to equalize terms of the form $\tapp {(\tlam x {t_1})} {t_2}$ with $\subst x {t_2} {t_1}$. It would work perfectly for the identity or substitution of closed terms. Let us define such a contraction $\contr \nv$ as follows $\tapp {(\tlam x {t_1})} {t_2} \;\contr \nv\; \subst x {t_2} {t_1}$, and see that identity $I$ indeed returns exactly what it obtains $\tapp I K = \tapp {(\tlam x x)} K \;\red {} \nv\; \subst x K x = K$, while the constant function of identity returns the identity regardless of its argument:
$$\tapp {\tapp K I} S =
\tapp {\tapp {(\tlam {xy} x)} I} S \;{\red {} \nv}\;
\tapp {\subst x I {(\tlam y x)}} S =
\tapp {(\tlam y I)} S \;\red {} \nv\;
\subst y S I = I$$
However, the constant function of a variable $y$ would not return $y$ because of the variable capture:
$$\tapp {\tapp K y} S =
\tapp {\tapp {(\tlam {xy} x)} y} S \;{\red {} \nv}\;
\tapp {\subst x y {(\tlam y x)}} S =
\tapp {(\tlam y x)} S \;\red {} \nv\;
\subst y S y = S \neq y$$

Moreover, we can easily see that $(\tlam x x)$ and $(\tlam y y)$ describe in practice the same function. Therefore, we want to define them as $\alpha$-convertible, that is equivalent up to $\alpha$-renaming. We obtain the conversion from the $\alpha$-contraction sufficient to rename bound variables:

\begin{center}\begin{prooftree}
  \hypo{x' \notin \fv (t) \cup \bv (t)}
\infer1{\tlam {x} t \; \contr \alpha \; \tlam {x'} {\subst {x} {x'} t} }
\end{prooftree}\end{center}

We can define operational semantics for the lambda calculus in the form of reduction semantics that specifies $\beta$-contraction as the single computation step.
We define $\beta$-contraction in the standard way
(and its contextual closure determines $\beta$-reduction):

\begin{center}\begin{prooftree}
  \hypo{t_1 \; \converts {} \alpha \; t'_1}
  \hypo{{\bv {(t'_1)}} \cap {\fv {(t_2)}} = \emptyset}
\infer2{\tapp {(\tlam x {t_1})} {t_2}
\; \contr \beta \;
\subst x {t_2} {t'_1}}
\end{prooftree}\end{center}
Thanks to the substitutability check, that is the second premiss of the rule above, the variable capture is avoided:
$$\tapp {\tapp K y} S =
\tapp {\tapp {(\tlam {xy} x)} y} S \;{\red {} \beta}\;
\tapp {\subst x y {(\tlam z x)}} S =
\tapp {(\tlam z x)} S \;\red {} \beta\;
\subst z S y = y$$
In short, the constant function constructor works as it should:
$\tapp {\tapp K y} S \;\converts {} \beta\; y$.

The definition can be simplified by \emph{capture-avoiding substitution} that generates fresh variable names on the fly for any bindings in $t_1$. With such an interpretation of the substitution symbol, the premisses in the $\beta$-contraction rule can be omitted.

\subsection{Nameless Representations of Functions}

Representation of functions with named variables entails a degree of freedom. For example, the identity function can be represented by $(\tlam x x)$ or $(\tlam y y)$, or a term with any other variable name. However, variable names can be abstracted out from closed terms in the following sense. There exists a translation function $\alpha$ such that $\alpha(t_1) = \alpha(t_2)$ if and only if $t_1 =_\alpha t_2$. Two possible representations are named after Nicolaas Govert de Bruijn, and they are de Bruijn indices and de Bruijn levels representations. Both of them use the following syntax of terms.
$$ \Lambda_\alpha \ni  t ::= n \alt \tapp{t_1}{t_2} \alt \lambda t $$

Lambda abstractions do not use variable names, and variables use natural numbers instead of names.
In the indices representation, variables identify their binder by giving the number of binders between them in the abstract syntax tree, while in the levels representation by giving the number of binders between the pointed binder and the root of the syntax tree.
Let us denote their translation functions as $\text{ind}: \Lambda_\text{closed} \to \Lambda_\alpha$ and $\text{lev}: \Lambda_\text{closed} \to \Lambda_\alpha$.
Then we have for example
$\text{ind}(\tlam x {\tapp x {\tlam y {\tapp x y}}}) = \lambda \tapp 0 {\lambda \tapp 1 0}$ and
$\text{lev}(\tlam x {\tapp x {\tlam y {\tapp x y}}}) = \lambda \tapp 0 {\lambda \tapp 0 1}$.

Some abstract machines work directly on de Bruijn representations. In case of strong reduction, it is required to shift de Bruijn indices or levels while performing the substitution
because the distance in the number of binders between a variable and its binder can change.
The appropriate shiftings are shown for example in \cite{DBLP:journals/lisp/Cregut07}.

\subsection{Reduction Strategies}
\label{sec:strategies}

The reduction order we have defined is nondeterministic. Specifically, the expression $(\const 1 \oplus \const 2) \oplus (\const 4 \oplus \const 8)$ can be reduced in two ways: reduce the sum $\const 1 \oplus \const 2$ first or $\const 4 \oplus \const 8$.
In order to impose a deterministic evaluation order, we can define a reduction strategy. We can do that by restricting the grammar of general contexts. Contextual closure under contexts $L ::= \hole \;|\; L \oplus e \;|\; \const n \oplus L$ gives left-to-right evaluation order, while under $R ::= \hole \;|\; R \oplus \const n \;|\; e \oplus R$ gives right-to-left order. We will denote restricted reduction with a nonterminal representing the restricted family of contexts over the arrow.
$$(\const 1 \oplus \const 2) \oplus (\const 4 \oplus \const 8) \;\red L +\;
\const 3 \oplus (\const 4 \oplus \const 8) \;\red L +\;
\const 3 \oplus \const {12} \;\red L +\;
\const{15}$$
$$(\const 1 \oplus \const 2) \oplus (\const 4 \oplus \const 8) \;\red R +\;
(\const 1 \oplus \const 2) \oplus \const {12} \;\red R +\;
\const 3 \oplus \const {12} \;\red R +\;
\const{15}$$
$$(\const 1 \oplus \const 2) \oplus (\const 4 \oplus \const 8) \;\nred L +\;
(\const 1 \oplus \const 2) \oplus \const {12} \phantom{\;\red R +\;
\const 3 \oplus \const {12} \;\red R +\;
\const{15}}$$

In the presence of functions, there are more interesting strategy choices. Informally, let us consider a function $f$ defined by the formula ``$f(x) = x + x$'' and the expression ``$f(3 + 7)$''.
We can do a step of the computation substituting the argument unevaluated and obtain an expression ``$(3 + 7) + (3 + 7)$''.
In the lambda calculus, this approach is called call by name.
Alternatively, we can evaluate the argument ``$3 + 7$'' to obtain the expression ``$f(10)$''
and only then perform the substitution to obtain ``$10 + 10$''.
This approach is called call by value.

Formally, call by name ($\mathit{cbn}$ of \cite{I}) step can be defined as contextual closure of $\beta$-contraction under the family of contexts $Q ::= \hole \;|\; \tapp Q t$.
In order to obtain a deterministic call-by-value reduction, we first need to restrict $\beta$-contraction to $\beta_\lambda$-contraction: $\tapp {(\tlam {x_1} {t_1})} {(\tlam {x_2} {t_2})} \;\contr {\beta_\lambda}\; \subst {x_1} {(\tlam {x_2} {t_2})} {t_1}$ under analogous assumptions as ordinary $\beta$-contraction.
Then left-to-right call-by-value reduction ($\mathit{lcbv}$ of \cite{I}) is obtained by a contextual closure under $E ::= \hole \;|\;\tapp E t\;|\; \tapp {(\tlam x t)} E$, while right-to-left ($\mathit{rcbv}$ of \cite{I}) under $F ::= \hole \;|\; \tapp F {(\tlam x t)}\;|\; \tapp t F$. Both call by name and call by value give rise to Turing-complete programming languages.

As visible in the example of ``$f(3 + 7)$'', the choice of the reduction strategy may affect the number of steps to achieve the final result. It may also happen that a computation terminates in one reduction strategy and diverges in another as in the case of $\tapp {\tapp K I} \Omega$ in call by name and call by value respectively:

\begin{center}
\begin{tabular}{cccccccc}
$\tapp {\tapp K I} \Omega$ & $\;\red Q {\beta_{\phantom{\lambda}}}\;$ &
$\tapp {(\tlam z I)}  \Omega$ & $\;\red Q {\beta_{\phantom{\lambda}}}\;$ &
$I$ & $\nred {} {\beta_{\phantom{\lambda}}}$\\
$\tapp {\tapp K I} \Omega$ & $\;\red E {\beta_\lambda}\;$ &
$\tapp {(\tlam z I)} \Omega$ &  $\;\red E {\beta_\lambda}\;$&
$\tapp {(\tlam z I)} \Omega$ &  $\;\red E {\beta_\lambda}\;$& \ldots\\
$\tapp {\tapp K I} \Omega$ & $\;\red F {\beta_\lambda}\;$ &
$\tapp {\tapp K I} \Omega$ &  $\;\red F {\beta_\lambda}\;$&
$\tapp {\tapp K I} \Omega$ &  $\;\red F {\beta_\lambda}\;$& \ldots\\
\end{tabular}
\end{center}

Nevertheless, even if two strategies part ways, the obtained expressions can be always brought back together by full $\beta$-reduction ($\mathit{full\beta}$ of \cite{I}), as guaranteed by the Church-Rosser theorem.

However, all the three reduction strategies from the paragraph above are weak strategies. This means that they do not reduce expressions under lambdas. For example, $\tlam x {\tapp {(\tlam y y)} x}$ is irreducible in any of the three but reduces in full beta-reduction: $\tlam x {\tapp {(\tlam y y)} x} \;\red {} \beta\; I$. 

An example of a strong strategy is normal-order reduction ($\mathit{no}$ of \cite{I}). It is obtained by contextual closure of ordinary $\beta$-contraction under leftmost-outermost contexts expressed by the nonterminal $N$.
\begin{alignat*}{10}
N &::= \overline{N} \alt \tlam x N&&&
n &::= & \; \tlam x n &\alt a\\
\overline{N} &::= \hole \alt \tapp {\overline{N}} t \alt \tapp a N&&\hspace{15mm}&
a &::= & \tapp a n &\alt x
\end{alignat*}

Normal-order reduction extends the call-by-name strategy, hence it is sometimes called strong call by name. All technical chapters of this dissertation consider strong reduction strategies built over weak ones.

\subsection{Abstract Machines for Additive Expressions}

Lambda calculus is a relatively convenient computational model to program in it directly. On the other hand, computational complexity issues are much less evident than in the model of Turing machines. Even a single $\beta$-reduction step requires searching for a contraction's context and performing the substitution while evading variable captures. These are nonlocal operations whose costs are nonobvious. However, this problem can be avoided by specifying the cost model for reduction via other abstract machines. Their desired property, enjoyed by Turing machines by design, is that they can always perform one step of computation in a fixed measure of time. Then the number of machine steps is a realistic cost model of time complexity.

We call these machines abstract because they abstract from details of a physical computer. They could also be called theoretical or logical because they are formally defined mathematical objects.

As an example, let us consider an abstract machine for additive expressions. It is presented in Figure~\ref{fig:machine}.
Its configuration consists of a focused expression, a stack representing evaluation contexts, and a direction tag that informs if the focused expression is already evaluated. Their grammar is given in the top part of the figure. A computation starts with the initial configuration focused on the root of the given expression $e$, with the empty context represented by $\nil$, and tag directed downwards: $\etrian$. The computation step consists in applying one of the four transition rules whose left-hand side matches the current configuration. The labels on the right are names of transitions used for convenience. The computation terminates if none of the rules applies. The result is then stored as the focused expression, in the empty context with the tag directed upwards: $\ctrian$.

\begin{figure}[h!]
\caption{Left-to-right abstract machine for additive expressions}
\label{fig:machine}
\vspace{-6mm}
\begin{alignat*}{3}
\textit{\textrm{Expressions}} \ni && e &:= \const n \alt e_1 \oplus e_2\\
\textit{\textrm{Frames}} \ni &\;& F &::= \hole \oplus e \alt \const n \oplus \hole\\
\textit{\textrm{Stacks}} \ni && S &::= \nil \alt \cons F S\\
\textit{\textrm{Configurations}} \ni && k &::= \econf e S \alt \cconf {\const n} S\\
\textit{\textrm{Transitions:}} \phantom{\ni} &&
\end{alignat*}%
\vspace{-12mm}
\begin{align}
e &\mapsto \econf e \nil\nonumber\\
\econf {e_1 \oplus e_2} S &\to
\econf {e_1} {\cons {\hole \oplus e_2} S} \tag{$\leftstra$} \label{tr:1}\\
\econf {\const n}  S &\to
\cconf {\const n}  S \tag{$\,\uparrow\,$} \label{tr:2}\\
\cconf {\const n}  {\cons {\hole \oplus e} S} &\to
\econf {e}  {\cons {\const n \oplus \hole} S} \tag{$\rightstra$} \label{tr:3}\\
\cconf {\const {n_2}}  {\cons {\const {n_1} \oplus \hole} S} &\to
\cconf {\const {n_1 + n_2}}  S \tag{$+$} \label{tr:4}\\
\cconf {\const n} \nil &\mapsto n\nonumber
\end{align}
\end{figure}

An example run of the machine is shown in Figure~\ref{fig:run}.
The first two steps of the computation to the same transition rule (\ref{tr:1})
and put the frames of the context ${({\hole} \oplus {\const 2})} \oplus {{({\const 4} \oplus {\const 8})}}$ of the subexpression $\const 1$ on the stack.
Subexpressions $\const 1$ and $\const 2$ are recognized as values in the third and the fifth step respectively
and are added in the sixth step.
After thirteen steps any of the transition rules cannot be applied, and the configuration matches the shape of a terminal configuration.
This way the initial expression has been reduced to $\const {15}$.

\begin{figure}[h!]
\caption{Example run of an abstract machine}
\label{fig:run}
\vspace{-6mm}
\begin{align*}
0:&& \langle {{{{({\const 1} \oplus {\const 2})}} \oplus {{({\const 4} \oplus {\const 8})}}}},&& {\nil}\rangle_\etrian &\stackrel{(\text{\ref{tr:1}})}{\to}\\[-0.5ex]
1:&& \langle {{{\const 1} \oplus {\const 2}}},&& {{\cons {{({\hole} \oplus {{({\const 4} \oplus {\const 8})}})}} {\nil}}}\rangle_\etrian &\stackrel{(\text{\ref{tr:1}})}{\to}\\[-0.5ex]
2:&& \langle {\const 1},&& {{\cons {{({\hole} \oplus {\const 2})}} {{\cons {{({\hole} \oplus {{({\const 4} \oplus {\const 8})}})}} {\nil}}}}}\rangle_\etrian &\stackrel{(\text{\ref{tr:2}})}{\to}\\[-0.5ex]
3:&& \langle {\const 1},&& {{\cons {{({\hole} \oplus {\const 2})}} {{\cons {{({\hole} \oplus {{({\const 4} \oplus {\const 8})}})}} {\nil}}}}}\rangle_\ctrian &\stackrel{(\text{\ref{tr:3}})}{\to}\\[-0.5ex]
4:&& \langle {\const 2},&& {{\cons {{({\const 1} \oplus {\hole})}} {{\cons {{({\hole} \oplus {{({\const 4} \oplus {\const 8})}})}} {\nil}}}}}\rangle_\etrian &\stackrel{(\text{\ref{tr:2}})}{\to}\\[-0.5ex]
5:&& \langle {\const 2},&& {{\cons {{({\const 1} \oplus {\hole})}} {{\cons {{({\hole} \oplus {{({\const 4} \oplus {\const 8})}})}} {\nil}}}}}\rangle_\ctrian &\stackrel{(\text{\ref{tr:4}})}{\to}\\[-0.5ex]
6:&& \langle {\const 3},&& {{\cons {{({\hole} \oplus {{({\const 4} \oplus {\const 8})}})}} {\nil}}}\rangle_\ctrian &\stackrel{(\text{\ref{tr:3}})}{\to}\\[-0.5ex]
7:&& \langle {{{\const 4} \oplus {\const 8}}},&& {{\cons {{({\const 3} \oplus {\hole})}} {\nil}}}\rangle_\etrian &\stackrel{(\text{\ref{tr:1}})}{\to}\\[-0.5ex]
8:&& \langle {\const 4},&& {{\cons {{({\hole} \oplus {\const 8})}} {{\cons {{({\const 3} \oplus {\hole})}} {\nil}}}}}\rangle_\etrian &\stackrel{(\text{\ref{tr:2}})}{\to}\\[-0.5ex]
9:&& \langle {\const 4},&& {{\cons {{({\hole} \oplus {\const 8})}} {{\cons {{({\const 3} \oplus {\hole})}} {\nil}}}}}\rangle_\ctrian &\stackrel{(\text{\ref{tr:3}})}{\to}\\[-0.5ex]
10:&& \langle {\const 8},&& {{\cons {{({\const 4} \oplus {\hole})}} {{\cons {{({\const 3} \oplus {\hole})}} {\nil}}}}}\rangle_\etrian &\stackrel{(\text{\ref{tr:2}})}{\to}\\[-0.5ex]
11:&& \langle {\const 8},&& {{\cons {{({\const 4} \oplus {\hole})}} {{\cons {{({\const 3} \oplus {\hole})}} {\nil}}}}}\rangle_\ctrian &\stackrel{(\text{\ref{tr:4}})}{\to}\\[-0.5ex]
12:&& \langle {\const {12}},&& {{\cons {{({\const 3} \oplus {\hole})}} {\nil}}}\rangle_\ctrian &\stackrel{(\text{\ref{tr:4}})}{\to}\\[-0.5ex]
13:&& \langle {\const {15}},&& {\nil}\rangle_\ctrian
&\stackrel{\phantom{(\text{\ref{tr:1}})}}{\nrightarrow}
\end{align*}
\vspace{-6mm}
\end{figure}

Each configuration of the machine can be decoded into a decomposition of the current expression into a focused subexpression and its context. A decomposition can be recomposed into the expression by plugging the subexpression into the context. The decoding with examples is given in Figure~\ref{fig:decoding}.
Formally, there are defined three decoding functions whose domains are stacks, configurations, and decompositions.
However, here the operator of underlining is overloaded and denotes each of the three functions.
Note that the double underlining can be used to decode configuration into recomposed expression.
The first example shows the decoding of the stack ${{\cons {{({\hole} \oplus {{({\const 4} \oplus {\const 8})}})}} {\nil}}}$ into the context ${{\hole} \oplus {{({\const 4} \oplus {\const 8})}}}$ step by step.
The second example shows that the configuration after the sixth step of the execution from Figure~\ref{fig:run} represents the expression ${{{\const 3} \oplus {{({\const 4} \oplus {\const 8})}}}}$.

\begin{figure}[h!]
\caption{Decoding of the left-to-right abstract machine for additive expressions}
\label{fig:decoding}\vspace{-6mm}
\begin{alignat*}{3}%
\underline{\argfill} & : \textit{Stacks} \to \textit{Contexts}\\
\underline{\nil} & = \hole \hspace{16mm}
\underline{\cons {(\hole \oplus e)} S}  = \plug {\underline{S}} {\hole \oplus e}\hspace{16mm}
\underline{\cons {(\const n \oplus \hole)} S}  = \plug {\underline{S}} {\const n \oplus \hole}\\
\underline{\argfill} & : \textit{Configurations} \to \textit{Expressions} \times \textit{Contexts}\\
\underline{\econf e S} & = (e, \underline S) \hspace{19mm}
\underline{\cconf {\const n} S}   = (\const n, \underline S)\\
\underline{\argfill} & : \textit{Expressions} \times \textit{Contexts} \to \textit{Expressions}\\
\underline{(e, G)} & = \plug G e
\end{alignat*}%
\begin{alignat*}{2}
\underline{\langle {{{\const 1} \oplus {\const 2}}}, {{\cons {{({\hole} \oplus {{({\const 4} \oplus {\const 8})}})}} {\nil}}}\rangle_\etrian} &=
({{{\const 1} \oplus {\const 2}}}, \underline{{\cons {{({\hole} \oplus {{({\const 4} \oplus {\const 8})}})}} {\nil}}})\\
&=
({{{\const 1} \oplus {\const 2}}}, \underline{\nil}[{{\hole} \oplus {{({\const 4} \oplus {\const 8})}}}])\\
&=
({{{\const 1} \oplus {\const 2}}}, \hole[{{\hole} \oplus {{({\const 4} \oplus {\const 8})}}}])\\
&=
({{{\const 1} \oplus {\const 2}}}, {{\hole} \oplus {{({\const 4} \oplus {\const 8})}}})
\end{alignat*}
$$\underline{({{{\const 1} \oplus {\const 2}}}, {{\hole} \oplus {{({\const 4} \oplus {\const 8})}}})} = {{{{({\const 1} \oplus {\const 2})}} \oplus {{({\const 4} \oplus {\const 8})}}}}$$
\begin{alignat*}{2}
\underline{\underline{\langle {\const 3}, {{\cons {{({\hole} \oplus {{({\const 4} \oplus {\const 8})}})}} {\nil}}}\rangle_\ctrian}}
&= \underline{(\const 3, {{\hole} \oplus {{({\const 4} \oplus {\const 8})}}})}\\
&= \plug {({{\hole} \oplus {{({\const 4} \oplus {\const 8})}}})} {\const 3}\\
&= {{{\const 3} \oplus {{({\const 4} \oplus {\const 8})}}}} 
\end{alignat*}
\vspace{-6mm}
\end{figure}

Influenced by the decoding of the configurations to decompositions, we can present the run from Figure~\ref{fig:run} as in Figure~\ref{fig:refocusing}.
The notation is called \emph{refocusing notation} in Chapter IV.
The name refers to Olivier Danvy's and Lasse R. Nielsen's concept of refocusing \cite{DBLP:journals/tcs/DanvyN01} which they connected to abstract machines in \cite{Danvy04refocusingin}.
The idea is that the machine looks at the given expression, and the only thing that transitions (\ref{tr:1}), (\ref{tr:2}), and (\ref{tr:3}) do is that they locally change the machine's focus between subexpressions
whereas the transition (\ref{tr:4}) performs the $+$-contraction.
In general, the focused expression of the machine is marked with the angular brackets with a current direction tag, and the context is reconstructed around the focused expression. 
For example, a configuration $\langle {\const {12}}, {{\cons {{({\const 3} \oplus {\hole})}} {\nil}}}\rangle_\ctrian$ is represented by ${3} \oplus {{\cut {\const {12}}_\ctrian}}$.
An analogous run for a more complex abstract machine with an explanation of each step is given in Table 2. of \cite{IV}.

\begin{figure}[t]
\caption{Example run of an abstract machine in refocusing notation}
\label{fig:refocusing}
\vspace{-6mm}
\begin{align*}
0:&& {\cut {{{({\const 1} \oplus {\const 2})}} \oplus {{({\const 4} \oplus {\const 8})}}}_\etrian} &\stackrel{(\text{\ref{tr:1}})}{\to}\\[-0.5ex]
1:&& {{{\cut {{{\const 1} \oplus {\const 2}}}_\etrian}} \oplus {{({\const 4} \oplus {\const 8})}}} &\stackrel{(\text{\ref{tr:1}})}{\to}\\[-0.5ex]
2:&& {{{({{\cut {\const 1}_\etrian}} \oplus {\const 2})}} \oplus {{({\const 4} \oplus {\const 8})}}} &\stackrel{(\text{\ref{tr:2}})}{\to}\\[-0.5ex]
3:&& {{{({{\cut {\const 1}_\ctrian}} \oplus {\const 2})}} \oplus {{({\const 4} \oplus {\const 8})}}} &\stackrel{(\text{\ref{tr:3}})}{\to}\\[-0.5ex]
4:&& {{{({\const 1} \oplus {{\cut {\const 2}_\etrian}})}} \oplus {{({\const 4} \oplus {\const 8})}}} &\stackrel{(\text{\ref{tr:2}})}{\to}\\[-0.5ex]
5:&& {{{({\const 1} \oplus {{\cut {\const 2}_\ctrian}})}} \oplus {{({\const 4} \oplus {\const 8})}}} &\stackrel{(\text{\ref{tr:4}})}{\to}\\[-0.5ex]
6:&& {{{\cut {3}_\ctrian}} \oplus {{({\const 4} \oplus {\const 8})}}} &\stackrel{(\text{\ref{tr:3}})}{\to}\\[-0.5ex]
7:&& {{3} \oplus {{\cut {{{\const 4} \oplus {\const 8}}}_\etrian}}} &\stackrel{(\text{\ref{tr:1}})}{\to}\\[-0.5ex]
8:&& {{3} \oplus {{({{\cut {\const 4}_\etrian}} \oplus {\const 8})}}} &\stackrel{(\text{\ref{tr:2}})}{\to}\\[-0.5ex]
9:&& {{3} \oplus {{({{\cut {\const 4}_\ctrian}} \oplus {\const 8})}}} &\stackrel{(\text{\ref{tr:3}})}{\to}\\[-0.5ex]
10:&& {{3} \oplus {{({\const 4} \oplus {{\cut {\const 8}_\etrian}})}}} &\stackrel{(\text{\ref{tr:2}})}{\to}\\[-0.5ex]
11:&& {{3} \oplus {{({\const 4} \oplus {{\cut {\const 8}_\ctrian}})}}} &\stackrel{(\text{\ref{tr:4}})}{\to}\\[-0.5ex]
12:&& {{3} \oplus {{\cut {\const {12}}_\ctrian}}} &\stackrel{(\text{\ref{tr:4}})}{\to}\\[-0.5ex]
13:&& {\cut {\const {15}}_\ctrian} &\stackrel{\phantom{(\text{\ref{tr:1}})}}{\nrightarrow}
\end{align*}
\vspace{-6mm}
\end{figure}

The decoding allows us to observe how a transition affects the evaluated expression. Rules (\ref{tr:1}), (\ref{tr:2}), and (\ref{tr:3}) leave the expression unchanged, so they are called overhead transitions. Since all stacks decode to a context derivable from the $L$ nonterminal, the rule (\ref{tr:4}) performs the $+$-contraction in $L$-context. Therefore the machine performs the left-to-right strategy $\red L +$ for additive expressions. These observations can be formalized as the following lemmas.

\begin{lemma} Transitions $(\text{\ref{tr:1}})$, $(\text{\ref{tr:2}})$, and $(\text{\ref{tr:3}})$ are overhead transitions:$$\forall\; \iota \in \{\text{\ref{tr:1}}, \text{\ref{tr:2}}, \text{\ref{tr:3}} \}.\;\; k \stackrel{(\iota)}{\to} k'
\;\;\implies\;\;
 \underline{\underline{k}} =  \underline{\underline{k'}}$$
\end{lemma}
\begin{lemma} Transition $(\text{\ref{tr:4}})$ performs $\red L +$ reduction: $
k \stackrel{(\text{\ref{tr:4}})}{\to} k'
\;\;\implies\;\;
 \underline{\underline{k}} \;\red L +\;  \underline{\underline{k'}}$.
\end{lemma}

Analogously, an abstract machine for right-to-left strategy $\red R +$ can be defined as in Figure~\ref{fig:rmachine}.
The first difference can be observed in the grammar of frames that corresponds to $R$-contexts.
A~difference is also in the first transition rule that focuses on the subexpression $e_2$ first.

\begin{figure}[h]
\caption{Right-to-left abstract machine for additive expressions}
\label{fig:rmachine}
\vspace{-6mm}
\begin{alignat*}{3}%
\textit{\textrm{Expressions}} \ni && e &:= \const n \alt e_1 \oplus e_2\\
\textit{\textrm{Frames}} \ni &\;& F &::= e \oplus \hole \alt \hole \oplus \const n\\
\textit{\textrm{Stacks}} \ni && S &::= \nil \alt \cons F S\\
\textit{\textrm{Configurations}} \ni && k &::= \econf e S \alt \cconf {\const n} S\\
\textit{\textrm{Transitions:}} \phantom{\ni} &&
\end{alignat*}%
\vspace{-12mm}
\begin{align}
e &\mapsto \econf e \nil\nonumber\\
\econf {e_1 \oplus e_2} S &\to
\econf {e_2} {\cons {e_1 \oplus \hole} S} \tag{$\rightstra$}\\
\econf {\const n}  S &\to
\cconf {\const n}  S \tag{$\,\uparrow\,$}\\
\cconf {\const n}  {\cons {e \oplus \hole} S} &\to
\econf {e}  {\cons {\hole \oplus \const n} S} \tag{$\leftstra$}\\
\cconf {\const {n_1}}  {\cons {\hole \oplus \const {n_2}} S} &\to
\cconf {\const {n_1 + n_2}} S \tag{$+$}\\
\cconf {\const n} \nil &\mapsto {\const n}\nonumber
\end{align}
\end{figure}

\newpage

\subsection{Abstract Machines for Lambda Calculus}
\label{sec:lambdamachines}
The first abstract machine for the lambda calculus is Landin's SECD abstract machine reproduced in Figure~\ref{fig:SECD} \cite{Landin64,DBLP:conf/ppdp/AgerBDM03}.
Transition names come from the original paper.
It is implemented in the software artefact accompanying Chapter IV titled \href{https://zenodo.org/record/6606302}{Abstract Machines Workshop}. The SECD machine performs the right-to-left call by value, that is $\stackrel{F}\to_{\beta_\lambda}$ reduction.

Configuration of the SECD machine, according to the name, consists of a stack, an environment, the code controlling the computation, and a dump.
A stack is just a list of evaluated terms, while an environment is a dictionary whose keys are identifiers, and whose values are evaluated terms.
The code component is a list of terms and application instructions.
The dump captures a complete configuration consisting of a stack, an environment, a code, and a next dump.

\begin{figure}[h!]
\caption{The pure part of the SECD abstract machine}
\label{fig:SECD}\vspace{-6mm}
\begin{alignat*}{3}%
\textit{\textrm{Terms}} \ni && t &:= x \alt \tapp {t_1} {t_2} \alt \tlam x t\\
\textit{\textrm{Control}} \ni && C &::= \nil \alt \cons t C \alt \cons {\text{ap}} C \\
\textit{\textrm{Values}} \ni && v &:= (\tlam x t, E)\\
\textit{\textrm{Environments}} \ni &\;& E &::= \nil \alt \cons {(x, v)}E\\
\textit{\textrm{Stacks}} \ni && S &::= \nil \alt \cons v S\\
\textit{\textrm{Dumps}} \ni && D &::= \nil \alt \cons {(S, E, C)} D\\
\textit{\textrm{Configurations}} \ni && k &::= \conf {S, E} {C, D}\\
\textit{\textrm{Transitions:}} \phantom{\ni} &&
\end{alignat*}%
\vspace{-12mm}
\begin{align}
t &\mapsto \conf {\nil, \nil} {\cons t \nil, \nil}\nonumber\\
\conf {S, E} {\cons x C, D} &\to
\conf {\cons {E(x)} S, E} {C, D} \tag{2a}  \label{tr:2a}\\
\conf {S, E} {\cons {(\tlam x t)} C, D} &\to
\conf {\cons {(\tlam x t, E)} S, E} {C, D} \tag{2b}  \label{tr:2b}\\
\conf {S, E} {\cons {\tapp {t_1} {t_2}} C, D} &\to
\conf {S, E} {\cons {t_2} {\cons {t_1} {\cons {\text{ap}} C}}, D} \tag{2d}\\
\conf {\cons {(\tlam x t, E')} {\cons v S}, E} {\cons {\text{ap}} C, D} &\to
\conf {\nil, \cons {(x, v)} E'} {\cons t \nil, \cons {(S, E, C)} D}\tag{2c1} \label{tr:2c1}\\
\conf {\cons v \_, \_} {\nil, \cons {(S, E, C)} D} &\to
\conf {\cons v {S}, E} {C, D} \tag{1}\\
\conf {\cons v \nil, \nil, \nil} \nil &\mapsto v\nonumber
\end{align}
\end{figure}

The only non-overhead transition rule is (\ref{tr:2c1}) which performs $\beta_\lambda$-contraction.
However, the substitution of the argument $v$ for parameter $x$ in the body $t$ is delayed by adding an entry $(x, v)$ to the environment $E'$ enclosing the abstraction $\tlam x t$.
The current value of substituted variables is retrieved by the transition rule (\ref{tr:2a}).
Thus, already processed terms should be always decoded in the appropriate environments.
Therefore, in the rule (\ref{tr:2b}), evaluated terms are paired with their environments that enclose their substituted variables.
As pairs, they constitute \emph{closures} that are used in other environment-based abstract machines.
Since SECD is the first abstract machine for the lambda calculus, its design is more complex and its decoding is more involved than of the next two machines below.

The dump component of the SECD machine could be simplified out \cite{DBLP:conf/ifip2/FelleisenF87,DBLP:conf/ifl/Danvy04,DBLP:conf/icfp/AccattoliBM14}. The simplified machine would resemble a right-to-left variant of the CEK machine.
The CEK machine \cite{DBLP:conf/ifip2/FelleisenF87} for the pure lambda calculus, performing left-to-right call-by-value reduction $\stackrel{E}\to_{\beta_\lambda}$, is presented in Figure~\ref{fig:CEK}.
The numbering of the rules coincides with the numbering of the rules from Table 1 of \cite{DBLP:conf/ifip2/FelleisenF87}.
The rule (\ref{tr:CEK:5}) is responsible for $\beta_\lambda$-contraction,
and stacks are straightforwardly decoded to $E$-contexts of Section~\ref{sec:strategies}.

\begin{figure}[h!]
\caption{The pure part of the CEK abstract machine}
\label{fig:CEK}
\vspace{-6mm}
\begin{alignat*}{3}%
\textit{\textrm{Terms}} \ni && t &:= x \alt \tapp {t_1} {t_2} \alt \tlam x t\\
\textit{\textrm{Values}} \ni && v &:= (\tlam x t, E)\\\textit{\textrm{Environments}} \ni &\;& E &::= \nil \alt \cons {(x, v)} E\\\textit{\textrm{Frames}} \ni &\;& F &::= \tapp \hole {(t, E)} \alt \tapp v \hole\\
\textit{\textrm{Stacks}} \ni && S &::= \nil \alt \cons F S\\
\textit{\textrm{Configurations}} \ni && k &::= \econf {(t, E)} S \alt \cconf v S\\
\textit{\textrm{Transitions:}} \phantom{\ni} &&
\end{alignat*}%
\vspace{-12mm}
\begin{align}
\setcounter{equation}{0}
t &\mapsto \econf {(t, \nil)} {\nil}\nonumber\\
\econf {(x, E)} S &\to
\cconf {E(x)} S \label{tr:CEK:1}\\ 
\econf {(\tlam x t, E)} S &\to
\cconf {(\tlam x t, E)} S\label{tr:CEK:2}\\
\econf {(\tapp {t_1} {t_2}, E)} S &\to
\econf {(t_1, E)} {\cons {\tapp \hole {(t_2, E)}} S}\label{tr:CEK:3}\\
\cconf v {\cons {\tapp \hole {(t_2, E)}} S} &\to
\econf {(t_2, E)} {\cons {\tapp v \hole} S}\label{tr:CEK:4}\\
\cconf v {\cons {\tapp {(\tlam x t, E)} \hole} S} &\to
\econf {(t, \cons {(x, v)} E)} S\label{tr:CEK:5}\\
\cconf v {\nil} &\mapsto v\nonumber
\end{align}
\end{figure}

The canonical machine performing call by name, that is $\stackrel{Q}\to_{\beta}$ reduction, is the well-known Krivine machine \cite{Krivine85} presented in Figure~\ref{fig:Krivine}. It operates on lambda expressions in de Bruijn indices representation.
The appropriate closure substituted for a variable is found by rules (\ref{tr:K:3}) and (\ref{tr:K:4}).
The unrestricted $\beta$-contraction is performed by the rule (\ref{tr:K:2}), since the rule (\ref{tr:K:1}) creates closures with an unevaluated term.

\begin{figure}[h!]
\caption{Krivine machine}
\label{fig:Krivine}
\vspace{-6mm}
\begin{alignat*}{3}%
\textit{\textrm{Terms}} \ni && t &:= n \alt \tapp {t_1} {t_2} \alt \lambda t\\
\textit{\textrm{Closures}} \ni && c &:= (t, E)\\\textit{\textrm{Environments}} \ni &\;& E &::= \nil \alt \cons c E\\
\textit{\textrm{Stacks}} \ni && S &::= \nil \alt \cons {\tapp \hole c} S\\
\textit{\textrm{Configurations}} \ni && k &::= \conf {c} S\\
\textit{\textrm{Transitions:}} \phantom{\ni} &&
\end{alignat*}%
\vspace{-12mm}
\begin{align}
\setcounter{equation}{0}
t &\mapsto \conf {(t, \nil)} {\nil}\nonumber\\
\conf {(\tapp {t_1} {t_2}, E)} S &\to
\conf {(t_1, E)} {\cons {\tapp \hole {(t_2, E)}} S}\label{tr:K:1}\\
\conf {(\lambda t, E)} {\cons {\tapp \hole c} S} &\to
\conf {(t, \cons c E)} S\label{tr:K:2}\\
\conf {(0, \cons c E)} S &\to
\conf {c} S\label{tr:K:3}\\
\conf {(n + 1, \cons c E)} S &\to
\conf {(n, E)} S\label{tr:K:4}\\
\conf c {\nil} &\mapsto c\nonumber
\end{align}
\end{figure}

Crégut extended the Krivine machine to a strong abstract machine called KN \cite{DBLP:journals/lisp/Cregut07}, recreated in Figure~\ref{fig:KN}.
The numbering of the transition rules coincides with \cite{DBLP:journals/lisp/Cregut07} where (7) is used to underline the shape of terminal configurations.
In case there is no argument on a stack for an abstraction, the rule (\ref{tr:KN:3}) is executed: an \emph{abstract variable} $V(n + 1)$ is added to the environment $E$ as a~prepared argument, and a frame $\lambda \hole$ is placed on the stack.
The number $n$, being the third element of the configuration, remembers the number of frames $\lambda \hole$ on the stack and therefore it is incremented by rule (\ref{tr:KN:3}) and decremented by rule (\ref{tr:KN:9}).
Thanks to that, an appropriate de Bruijn index can be easily computed by the rule (\ref{tr:KN:6}).

\begin{figure}[h!]
\caption{A variant of KN abstract machine}
\label{fig:KN}
\vspace{-6mm}
\begin{alignat*}{3}%
\textit{\textrm{Terms}} \ni && t &:= n \alt \tapp {t_1} {t_2} \alt \lambda t\\
\textit{\textrm{Closures}} \ni && c &:= (t, E) \alt V(n)\\\textit{\textrm{Environments}} \ni &\;& E &::= \nil \alt \cons c E\\
\textit{\textrm{Frames}} \ni &\;& F &::= \tapp \hole c \alt \tapp t \hole \alt \lambda \hole\\
\textit{\textrm{Stacks}} \ni && S &::= \nil \alt \cons F S\\
\textit{\textrm{Configurations}} \ni && k &::= \econf {c} {S, n} \alt \cconf t {S, n}\\
\textit{\textrm{Transitions:}} \phantom{\ni} &&
\end{alignat*}%
\vspace{-12mm}
\begin{align}
\setcounter{equation}{0}
t &\mapsto \econf {(t, \nil)} {\nil, 0}\nonumber\\
\econf {(\tapp {t_1} {t_2}, E)} {S, n} &\to
\econf {(t_1, E)} {\cons {\tapp \hole {(t_2, E)}} S, n} \label{tr:KN:1}\\
\econf {(\lambda t, E)} {\cons {\tapp \hole c} S, n} &\to
\econf {(t, \hspace{13mm} \cons c E)} {\phantom{\cons {\lambda \hole} {\,\,}} S, n} \label{tr:KN:2}\\
\econf {(\lambda t, E)} {\cancel {\cons {\tapp \hole c} {}} \;S, n} &\to
\econf {(t, \cons {V(n + 1)} E)} {\cons {\lambda \hole} S, n + 1} \label{tr:KN:3}\\
\econf {(0, \cons c E)} {S, n} &\to
\econf {c} {S, n} \label{tr:KN:4}\\
\econf {(n' + 1, \cons c E)} {S, n} &\to
\econf {(n', E)} {S, n} \label{tr:KN:5}\\
\econf {V(n_0)} {S, n} &\to
\cconf {n - n_0} {S, n} \label{tr:KN:6}\\
\setcounter{equation}{7}
\cconf t {\cons {\tapp \hole c} S, n} &\to
\econf c {\cons {\tapp t \hole} S, n} \label{tr:KN:8}\\
\cconf t {\cons {\lambda \hole} S, n} &\to
\cconf {\lambda t} {S, n - 1} \label{tr:KN:9}\\
\cconf {t_2} {\cons {\tapp {t_1} \hole} S, n} &\to
\cconf {\tapp {t_1} {t_2}} {S, n} \label{tr:KN:10}\\
\cconf t {\nil, 0} &\mapsto t\nonumber
\end{align}
\end{figure}

Crégut's KN machine performs $\stackrel N \to_\beta$ reduction.
However, the argument that KN implements the strong call-by-name strategy is not immediate \cite{DBLP:journals/jfp/Garcia-PerezN19}. One of the possible proofs is to give a weakly bisimilar abstract machine carrying extra information, called a ghost abstract machine in \cite{IV}. As an example, a version of KN with explicit shape invariant is given in Figure~\ref{fig:KNi}.
The grammar of stacks with nonterminals $\alpha$ and $\nu$ guards the property that $\lambda\hole$ is never placed directly on $\tapp \hole c$ because it would create an unreduced function-argument pair on the stack.
Such $\alpha$-stacks are decoded into exactly $N$-contexts of Section~\ref{sec:strategies}.
The grammar also uses grammars of terms explicitly restricted to normal and neutral terms.
The explicit coercion $\lceil \argfill \rceil$ from neutral to normal terms, including the $\lceil \hole \rceil$ frame, is erased by the decoding.
However, management of this explicit shape invariant increases the number of configuration modes to three and adds a transition rule (\ref{tr:KNi:8a}).
Yet, the side condition of transition rule (\ref{tr:KN:3}) from Figure \ref{fig:KN} is expressed explicitly by the presence of the $\lceil \hole \rceil$ frame.

\begin{figure}[h!]
\caption{Ghost abstract machine with explicit shape invariant for KN}
\label{fig:KNi}
\vspace{-6mm}
\begin{alignat*}{3}%
\textit{\textrm{Terms}} \ni && t &:= m \alt \tapp {t_1} {t_2} \alt \lambda t\\
\textit{\textrm{Normal Terms}} \ni && n &:= \lambda t \alt \lceil a \rceil\\
\textit{\textrm{Neutral Terms}} \ni && a &:= m \alt \tapp a n\\
\textit{\textrm{Closures}} \ni && c &:= (t, E) \alt V(n)\\\textit{\textrm{Environments}} \ni &\;& E &::= \nil \alt \cons c E\\
\mathit{Applicative\;Stacks}    \ni && \alpha &::=
         \cons {\tapp \hole c} \alpha
	\alt \cons {\lceil \hole \rceil} \nu \\
\mathit{Non\text{-}applicative\;Stacks} \ni && \nu    &::=
	     \nil
	\alt \cons {\lambda \hole} \nu
	\alt \cons {\tapp a \hole} \alpha \\
\mathit{Configurations} \ni && k   &::=
  \econf c \alpha
  \alt \cconf a \alpha
  \alt \nconf n \nu\\
\textit{\textrm{Transitions:}} \phantom{\ni} &&
\end{alignat*}%
\vspace{-12mm}
\begin{align}
\setcounter{equation}{0}
t &\mapsto \econf {(t, \nil)} {\cons {\lceil \hole \rceil} \nil, 0}\nonumber\\
\econf {(\tapp {t_1} {t_2}, E)} {\alpha, m} &\to
\econf {(t_1, E)} {\cons {\tapp \hole {(t_2, E)}} \alpha, m}\\
\econf {(\lambda t, E)} {\cons {\tapp \hole c} \alpha, m} &\to
\econf {(t, \hspace{13mm} \cons c E)} {\phantom{\cons {\lambda \hole} {\,\,}} \alpha, m}\\
\econf {(\lambda t, E)} {\cons {\lceil \hole \rceil} \nu, m} &\to
\econf {(t, \cons {V(m + 1)} E)} {\cons {\lambda \hole} \nu, m + 1}\\
\econf {(0, \cons c E)} {\alpha, m} &\to
\econf {c} {\alpha, m}\\
\econf {(m' + 1, \cons c E)} {\alpha, m} &\to
\econf {(m', E)} {\alpha, m}\\
\econf {V(m_0)} {\alpha, m} &\to
\cconf {m - m_0} {\alpha, m}\\
\setcounter{equation}{7}
\cconf a {\cons {\tapp \hole c} \alpha, m} &\to
\econf c {\cons {\lceil \hole \rceil} {\cons {\tapp a \hole} \alpha}, m}\\
\cconf a {\cons {\lceil \hole \rceil} \nu, m} &\to
\nconf {\lceil a \rceil}  {\nu, m} \tag{8a} \label{tr:KNi:8a}\\
\nconf n {\cons {\lambda \hole} \nu, m} &\to
\nconf {\lambda n} {\nu, m - 1}\\
\nconf {n} {\cons {\tapp a \hole} \alpha, m} &\to
\cconf {\tapp a n} {\alpha, m}\\
\cconf n {\nil, 0} &\mapsto n\nonumber
\end{align}
\end{figure}

It turns out that KN may need exponentially many steps with respect to simulated reduction steps.
As Accattoli and Dal Lago pinpointed with the name \emph{size explosion problem} \cite{DBLP:journals/corr/AccattoliL16},
a lambda expression can grow exponentially big with respect to the number of executed $\beta$-steps.
Since the KN employs no means of sharing, exponentially many constructors have to be introduced then step by step.
In other words, in general, KN suffers from exponential overhead. The first abstract machine for strong call by name with polynomial overhead was Accattoli's Useful MAM \cite{DBLP:conf/wollic/Accattoli16}, reproduced in Figure~\ref{fig:UMAM}.
It represents contexts of focused expressions by structures called here stacks and dumps in analogy to the SECD machine.
It also uses a global store $\sigma$ instead of local environments.
Variables are intended to be implemented as memory locations, so the store is represented by the memory state.
Due to that, it needs to maintain an invariant that all variables have distinct names. It can be achieved by generating fresh variable names in the initial configuration and transition rules (\ref{tr:UMAM:er}) and (\ref{tr:UMAM:ea}).
Useful MAM implements \emph{useful sharing} \cite{DBLP:journals/corr/AccattoliL16} by suppressing substitutions that do not create new reducible expressions as in the transition rule (\ref{tr:UMAM:c3}).
The $\beta$-contraction is realized by rules (\ref{tr:UMAM:m1}) and (\ref{tr:UMAM:m2}).
The rule (\ref{tr:UMAM:m1}) eagerly executes substitution, which is of linear complexity, in order to avoid creating \emph{variable chains} in the store.
The rule (\ref{tr:UMAM:m2}) computes an appropriate label $l$ with an auxiliary abstract machine with the same data structures and \texttt{c}-transitions.
The label $red$ means that the term $t_2$ contains a redex and its substitution is always useful, $neu$ means that $t_2$ is a neutral term and its substitution would be always useless, and $abs$ means that $t_2$ is a normal abstraction and its substitution is useful if and only if it will be applied to an argument.

\begin{figure}[h!]
\caption{Useful Milner abstract machine}
\label{fig:UMAM}
\vspace{-6mm}
\begin{alignat*}{3}%
\textit{\textrm{Terms}} \ni && t &:= x \alt \tapp {t_1} {t_2} \alt \tlam x t\\
\textit{\textrm{Stacks}} \ni && S &::= \nil \alt \cons {\tapp \hole t} S\\
\textit{\textrm{Dump Frames}} \ni &\;& F &::= \tlam x \hole \alt (\tapp t \hole) S\\
\textit{\textrm{Dumps}} \ni && D &::= \nil \alt \cons F D\\
\textit{\textrm{Labels}} \ni && l &:= abs \alt neu \alt red\\\textit{\textrm{Stores}} \ni &\;& \sigma &::= \nil \alt \cons {(x, (t, l))}  \sigma\\
\textit{\textrm{Configurations}} \ni && k &::= \econf {t} {S, D, \sigma} \alt \cconf t {S, D, \sigma}\\
\textit{\textrm{Transitions:}} \phantom{\ni} &&
\end{alignat*}%
\vspace{-12mm}
\begin{align}
\setcounter{equation}{0}
t &\mapsto \econf {t^\alpha} {\nil, \nil, \nil}\nonumber\\
\econf {\tapp {t_1} {t_2}} {S, D, \sigma} &\to
\econf {t_1} {\cons {\tapp \hole {t_2}} S, D, \sigma} \tag{$\etrian\texttt{c}_1$}\\
\econf {\tlam {x_1} t} {\cons {\tapp \hole {x_2}} S, D, \sigma} &\to
\econf {\subst {x_1} {x_2} t} {S, D, \sigma} \tag{$\texttt{m}_1$} \label{tr:UMAM:m1}\\
\econf {\tlam {x} {t_1}} {\cons {\tapp \hole {t_2}} S, D, \sigma} &\to
\econf {t_1} {S, D, \cons {(x, (t_2, l))} \sigma} \tag{$\texttt{m}_2$}\label{tr:UMAM:m2} &&: \text{appropriate $l$}\\
\econf {\tlam {x} t} {\nil, D, \sigma} &\to
\econf {t} {\nil, \cons {\tlam x \hole} D, \sigma} \tag{$\etrian\texttt{c}_2$}\\
\econf {x} {S, D, \sigma} &\to
\econf {t^\alpha} {S, D, \sigma} \tag{$\texttt{e}_{red}$} \label{tr:UMAM:er}&&: \sigma(x)= (t, red)\\
\econf {x} {\cons {\tapp \hole {t_2}} S, D, \sigma} &\to
\econf {t^\alpha} {\cons {\tapp \hole {t_2}}S, D, \sigma} \tag{$\texttt{e}_{abs}$} \label{tr:UMAM:ea}&&: \sigma(x)= (t, abs)\\
\econf {x} {S, D, \sigma} &\to
\cconf {x} {S, D, \sigma} \tag{$\etrian\texttt{c}_3$} \label{tr:UMAM:c3}&&: \text{otherwise}\\
\cconf {t_1} {\cons {\tapp \hole {t_2}} S, D, \sigma} &\to
\econf {t_2} {\nil, \cons {(\tapp {t_1} \hole) S} D, \sigma} \tag{$\ctrian\texttt{c}_6$}\\
\cconf {t} {\nil, \cons {\tlam x \hole} D, \sigma} &\to
\cconf {\tlam x t} {S, D, \sigma} \tag{$\ctrian\texttt{c}_4$}\\
\cconf {t_2} {\nil, \cons {(\tapp {t_1} \hole) S} D, \sigma} &\to
\cconf {\tapp {t_1} {t_2}} {S, D, \sigma} \tag{$\ctrian\texttt{c}_5$}\\
\cconf t {\nil, \nil, \sigma} &\mapsto (t, \sigma)\nonumber
\end{align}
\end{figure}

Accattoli, Condoluci, and Sacerdoti Coen proposed an abstract machine for strong call by value with polynomial overhead, independent from ours, and called it SCAM \cite{DBLP:conf/lics/AccattoliCC21}. It is reproduced for comparison in Figure~\ref{fig:SCAM} in a notation akin to ours. A~brief comparison of SCAM with RKLV of Chapter III is located at the end of that chapter.

\begin{figure}[h!]
\caption{Strong crumbling abstract machine}
\label{fig:SCAM}
\vspace{-6mm}
\begin{alignat*}{3}
\textit{\textrm{Bites}} \ni && b &::= x \alt \tapp x y \alt \tlam x t\\
\textit{\textrm{Environments}} \ni &\;& E &::= \hole \alt E \es x b\\
\hspace{-8mm}\textit{\textrm{Crumbles}} \ni && t &::= \plug E \star\\
\textit{\textrm{Frames}} \ni &\;& F &::= \hole \es x b \alt t \es x {\tlam y \hole}\\
\textit{\textrm{Stacks}} \ni && S &::= \nil \alt \cons F S\\
\textit{\textrm{Configurations}} \ni && k &::= \econf t S \alt \cconf t S\\
\textit{\textrm{Transitions:}} \phantom{\ni} &&
\end{alignat*}%
\vspace{-12mm}
\begin{align}
t &\mapsto \econf t \nil\nonumber\\
\econf {t\es x {\tapp y z}} S &\to
\econf {\plug {E} {t\es x b}{\subst w z {}} } S
&& : \text{$(*)$ $\wedge$ $S(z)$ is an abstr.} \tag{$\beta_v$} \label{tr:SCAM:bv}\\
\econf {t\es x {\tapp y z}} S &\to
\econf {\plug E {t\es x b}} {\cons {\hole \es w z} S}\hspace{-4.5mm}
&& : \text{$(*)$ $\wedge$ $S(z)$ not an abstr.} \tag{$\beta_i$} \label{tr:SCAM:bi}\\
\econf {t \es x y} S &\to \econf {\subst x y t} S
&& : x \neq \star \tag{$\mathsf{ren}$} \label{tr:SCAM:ren}\\
\econf {t \es x b} S &\to \econf t {\cons {\hole \es x b} S}
&& : \text{otherwise} \tag{$\mathsf{sea}_1$}\\
\econf \star S &\to \cconf \star S
\tag{$\mathsf{sea}_2$}\\
\cconf t {\cons {\hole \es x b} S} &\to \cconf {t \es x b} S
&& : \text{$b$ is\! not\! an\! abstraction} \tag{$\mathsf{sea}_3$}\\
\cconf t {\cons {\hole \es x  {\tlam y {t'}}} S} &\to \cconf t S
&& : x \notin \mathit{FV}(t) \tag{$\mathsf{gc}$}\\
\cconf t {\cons {\hole \es x {\tlam y {t'}}} S} &\to
\econf {t'} {\cons {t \es x {\tlam y \hole}} S}
&& : x \in \mathit{FV}(t) \tag{$\mathsf{sea}_5$}\\
\cconf {t'} {\cons {t \es x {\tlam y \hole}} S} &\to
\cconf {t \es x {\tlam y {t'}}} S
\tag{$\mathsf{sea}_4$}\\
\cconf t \nil &\mapsto t\nonumber
\end{align}

\end{figure}

The SCAM works on precompiled forms of lambda terms called \emph{crumbles} that resemble $A$-normal forms \cite{DBLP:journals/lisp/SabryF93} (later called administrative normal forms, for example in \cite{Appel98}).
They are represented by explicit substitutions $\argfill \es \argfill \argfill$ that are part of the syntax,
as opposed to substitution operation $\subst \argfill \argfill \argfill$ that is executed by rules (\ref{tr:SCAM:bv}) and (\ref{tr:SCAM:ren}).
The machine performs on-the-fly $\alpha$-renaming on $\beta$-transitions (\ref{tr:SCAM:bv}) and (\ref{tr:SCAM:bi}).
The side condition $(*)$ denotes $S(y)^\alpha = \tlam w {\plug E {\star \es \star b}}$ where $S(y)$ is an abstraction bound to the variable $y$ on the stack $S$.

\subsection{Normalization by Evaluation}

A term irreducible in the given strategy is called a normal form of this strategy. Thus the process of computing it is called normalization. A possible means of normalization is a simulation of reduction steps via an abstract machine, as presented in the previous section.

Another approach is to assign meaning to expressions in a way consistent with reduction rules.
Intuitively, it means that if two expressions $e_1$ and $e_2$ have the same meaning $\llbracket e_1 \rrbracket = \llbracket e_2 \rrbracket$, then they are convertible (in case of additive expressions: $e_1 =_+ e_2$).
Reification $\lceil \argfill \rceil$ is an operation that can turn a meaning $m$ into an expression $\lceil m \rceil$ having this meaning $\llbracket \lceil m \rceil \rrbracket = m$.
Having that, in order to normalize an expression~$e$, we can determine its meaning $\llbracket e \rrbracket$, and reify it into an expression  $\lceil \llbracket e \rrbracket \rceil$ in normal form that is convertible to $e$ (for additive expressions: $e =_+ \lceil \llbracket e \rrbracket \rceil$).
The meaning of additive expressions is defined compositionally as follows:
\begin{alignat*}{1}
\llbracket \argfill \rrbracket &: \mathcal E \to \mathbb N\\
\llbracket \const n \rrbracket &= n\\
\llbracket e_1 \oplus e_2 \rrbracket &= \llbracket e_1 \rrbracket + \llbracket e_2 \rrbracket
\end{alignat*}

Formally, it can be implemented in the OCaml language as in Listing~\ref{lst:normalizer}.\\
\begin{lstlisting}[caption={Normalizer for additive expressions},label={lst:normalizer}]
type expr  = Const of int | Plus of expr * expr
type value = int

let rec eval : expr -> value = function
  | Const n      -> n
  | Plus(e1, e2) -> (eval e1) + (eval e2)

let reify    (n : value) : expr = Const n
let normalize (e : expr) : expr = reify (eval e)
\end{lstlisting}

It can be checked that $\const {\llbracket(\const 1 \oplus \const 2) \oplus (\const 4 \oplus \const 8)\rrbracket} = \const {15}$: a metaexpression \texttt{normalize (Plus (Plus (Const 1) (Const 2)) (Plus (Const 4) (Const 8))))} computes a metavalue \texttt{Const 15}. 

\subsection{Functional Correspondence}

Olivier Danvy and his then doctoral students showed that such normalizers and abstract machines are two sides of the same coin \cite{DBLP:conf/ppdp/AgerBDM03}. They elaborated on how the two can be systematically transformed into each other, preserving their behaviour, with examples including SECD, CEK, and Krivine machine. The bridge between the two sides is called the \emph{functional correspondence}.

As an extra example, the transformation of the normalizer from Listing~\ref{lst:normalizer} is presented below. The first step is to translate the code to the continuation-passing style where the control flow is explicitized by passing a continuation as an extra argument of functions. The result of the translation of the function eval is presented in Listing~\ref{lst:cps}. Additionally, the left-to-right evaluation order is imposed, and the reify function is inlined and passed as a continuation in the normalize function where eval in the continuation-passing style is called.

\begin{lstlisting}[caption={Normalizer in continuation-passing style},label={lst:cps}]
let rec eval (t : expr) (k : value -> 'a) : 'a =
  match t with
  | Const n      -> k n
  | Plus(e1, e2) -> eval e1 (fun n1 ->
                    eval e2 (fun n2 ->
                    k (n1 + n2)))

let normalize (e : expr) : expr =
  eval e (fun n -> Const n)

\end{lstlisting}

The second step is to eliminate the higher-order aspect of function eval by defunctionalization. The type of continuation is defunctionalized into a datatype of the stack. The result is given in Listing~\ref{lst:defun}. Additionally, the definition of the stack is dissected into the definition of frames and stack as a list of frames.

\begin{lstlisting}[caption={Defunctionalized normalizer in continuation-passing style},label={lst:defun}]
type frame = PlusR of expr | PlusL of value
type stack = frame list

let rec eval (e : expr) (s : stack) : expr =
  match e with
  | Plus(e1, e2)  -> eval e1 (PlusR e2 :: s)
  | Const n       -> cont n s
and cont (n : value) (s : stack) : expr =
  match s with
  | PlusR e2 :: s -> eval e2 (PlusL n :: s)
  | PlusL n1 :: s -> cont (n1 + n) s
  | []            -> Const n

let normalize (e : expr) : expr = eval e []
\end{lstlisting}

Commonly, a code in the form of Listing~\ref{lst:defun} is regarded as an abstract machine from Figure~\ref{fig:machine}. The \texttt{eval t []} from the last line defines the initial configuration, and the last clause of pattern matching in the \texttt{cont} function defines how to retrieve the result from the terminal configuration.
The remaining four pattern matching clauses encode the four transitions of the machine, and names of mutually recursive functions \texttt{eval} and \texttt{cont}, with only tail calls, correspond to downward $\etrian$ and upward $\ctrian$ directions of the tag respectively.
However, the structure of abstract machines can be made even more explicit as in Listing~\ref{lst:machine}. Configurations form a datatype, \texttt{load} returns an initial configuration, \texttt{trans\_err} is the transition function implemented as a partial function from configurations to configurations, and \texttt{unload} manages the terminal configuration. The computation of normal form is directly expressed as an iteration of the finite transitions.

\begin{lstlisting}[caption={Abstract machine},label={lst:machine}]
type conf = E of expr * stack | C of value * stack

let load (e : expr) : conf = E(e, [])

let trans_err (c : conf) : conf =
  match c with
  | E(Plus(e1, e2),   s) -> E(e1, PlusR e2 :: s)
  | E(Const n,        s) -> C(n,              s)
  | C(n,  PlusR e  :: s) -> E(e,  PlusL  n :: s)
  | C(n2, PlusL n1 :: s) -> C(n1 + n2,        s)
  | C(n,             []) -> raise Stream.Failure

let unload : conf -> expr = function
  | C(n,             []) -> Const n
  | _                    -> assert false

let rec iter_trans (c : conf) : conf =
  match trans_err c with
  | exception Stream.Failure -> c
  | c'                       -> iter_trans c'

let normalize (e : expr) : expr =
  e |> load |> iter_trans |> unload
\end{lstlisting}

Normalization by evaluation is sometimes promoted as a “reduction-free” approach to normalization. However, thanks to Danvy’s bridge \cite{DBLP:journals/entcs/Danvy05}, it can be thought of as a reduction-hidden-at-the-metalevel approach. An abstract machine derived via the functional correspondence allows us to observe how the normalization by evaluation proceeds. The CPS translation makes the control flow, and maybe more importantly, the evaluation order explicit. In turn, the defunctionalization pushes the implementation of closures from the metalevel to the object level. Then decodings of the configurations enable tracing of the reduction sequence.

An obtained machine exhibits steps of, possibly higher-order, computation as iterated finite changes of first-order data. Such a representation constitutes quite clear implementation guidelines of the defined language. Moreover, the implementation is actually in trampolined style \cite{DBLP:conf/icfp/GanzFW99} as a single transition corresponds to one bounce. This means that a computation can be easily run for a given number of steps, resumed later, or interleaved with other computations. On the other hand, from a theoretical point of view, an abstract machine constitutes a operational semantics. It is more fine-grained than small-step operational semantics, or even than micro-step operational semantics as understood in \cite{DBLP:conf/icfp/AccattoliBM14}, so we can call it nano-step operational semantics. Together with a practical implementation, it has also a didactic value because learners can run an abstract computation step by step on their computers and see how the defined programming language works. Last but not least, the number of machine steps gives a natural cost model for time complexity, as mentioned earlier.

\subsection{Efficiency of Abstract Machines}

Chapters III and IV tackle problems of the efficiency of abstract machines.
The KNV machine of Chapter II suffers from exponential overhead because it may fall prey to computing the same partial results repeatedly.
It resembles the trouble with computing terms of the Fibonacci sequence directly from its recursive definition:

\begin{lstlisting}
let rec fib (n : int) : int =
  if n < 2 then n else fib (n - 1) + fib (n - 2)
\end{lstlisting}

A computer user can convince themselves empirically that the initial terms are computed immediately, but a query for \texttt{fib(60)} causes the computer to freeze for a~long while. The reason is that preceding sequence terms are computed over and over. The problem can be solved by the memoization technique:

\begin{lstlisting}
module IntDict = Map.Make(
  struct type t = int let compare = compare end)

let caches : int IntDict.t ref = ref IntDict.empty

let rec fib (n : int) : int =
  match IntDict.find_opt n !caches with
  | None   -> let y = if n < 2 then n else
                fib (n - 1) + fib (n - 2) in
              caches := IntDict.add n y !caches;
              y
  | Some y -> y
\end{lstlisting}

When a computation of any term is finished, the result is cached, and then retrieved from the cache any time it is needed again. Thanks to that, any term is computed at most once. The value of \texttt{fib(150)}, modulo OCaml’s integer overflow, is returned in the blink of an eye. The difference is truly astronomic because the number of computation steps directly from the recursive definition is greater than the number of nanoseconds since the Big Bang. If my computer had begun that computation at the very moment of the Big Bang, dinosaurs on Earth would have evolved before, non-avian dinosaurs would have become extinct, and my computer would be still computing now!

The introduction of memoization to normalizers of lambda terms entails the \emph{implicit sharing} of subterms. It leads to a counterintuitive proposition that exponentially big terms can be constructed in linear time. Since terms are represented as persistent data structures, when a computed subterm is needed again, it is enough to copy a memory reference to its root instead of copying every memory cell that represents it. Due to that, a logical size of the term can nearly double in constant time, while the physical representation grows at most by a constant amount of space. While it does not affect reasoning about time complexity, the fact that subterm sharing is managed by the metalanguage hampers the reasoning about space complexity. Therefore, a systematic explicitization of memory sharing is one of the subjects of still ongoing development.

As shown in Listing~\ref{lst:machine}, whole configurations of abstract machines
can be considered as persistent data structures. Thus, to bound the number of steps the machines can perform, potential functions à la Chris Okasaki were employed \cite{Okasaki:99}. As an example of how they are working, a potential function $\Phi$ for the abstract machine from Figure~\ref{fig:machine} is presented in Figure~\ref{fig:potential} with its properties expressed below.

\begin{figure}[h!!]
\caption{Potential function for the left-to-right abstract machine for additive expressions}
\label{fig:potential}
\vspace{-6mm}
\begin{alignat*}{3}%
&\Phi : \textit{Expressions} \to \mathbb{N}\\
&\Phi(e_1 \oplus e_2) = 3 + \Phi(e_1) + \Phi(e_2)\hspace{19mm}
\Phi(\const n) = 1\\
&\Phi : \textit{Stacks} \to \mathbb{N}\\
&\Phi(\nil) = 0 \hspace{12mm}
\Phi{(\cons {\hole \oplus e} S)}  = 2 + \Phi(e) + \Phi(S) \hspace{12mm}
\Phi{(\cons {\const n \oplus \hole} S)}  = 1 + \Phi(S)\\
&\Phi : \textit{Configurations} \to \mathbb{N}\\
&\Phi{(\econf e S)} = \Phi(e) + \Phi(S)\hspace{19mm}
\Phi{(\cconf {\const n} S)}   = \Phi(S)
\end{alignat*}%
\end{figure}

\begin{lemma}Every transition of the machine from Figure~\ref{fig:machine} decreases the potential:\\
$k \to k' \;\;\implies\;\;
 \Phi(k) >  \Phi(k')$.
\;\; More specifically:
$k \to k' \;\;\implies\;\;
 \Phi(k) = 1 + \Phi(k')$.
\end{lemma}

\begin{proof}Case analysis on transition rules.
\begin{align}
3 + \Phi(e_1) + \Phi(e_2) + \Phi(S) &>
\Phi(e_1) + 2 + \Phi(e_2) + \Phi(S) \tag{\ref{tr:1}}\\
1 + \Phi(S) &>
\Phi(S) \tag{\ref{tr:2}}\\
2 + \Phi(e) + \Phi(S) &>
\Phi(e) + 1 + \Phi(S) \tag{\ref{tr:3}}\\
1 + \Phi(S) &>
\Phi(S) \tag{\ref{tr:4}}
\end{align}
\end{proof}
\begin{theorem}Execution of an additive expression $e$ takes exactly $\Phi(e)$ steps.
\end{theorem}
\begin{proof}
The potential of the initial configuration is $\Phi(e)$, every transition decreases it by~1, and the machine stops if and only if the potential drops to 0.
\end{proof}

\subsection{Contributions}

The contributions of this dissertation are spread over the four technical chapters and summarized below. Table~\ref{tab:machines} gives the context of the research area by listing abstract machines for lambda calculus discussed in Section~\ref{sec:lambdamachines}.
The first three columns present respectively their names, articles that introduce each, and year of the publication.
Three of the machines (KNV, RKNV, RKNL) are contributions of the last three chapters of this dissertation.
The fourth column identifies reduction strategies implemented by the machines with strategies surveyed in Chapter I.
The machine called SCAM implements a strategy in a calculus with explicit substitutions, so its strategy is absent in \cite{I}.
However, the survey is planned to be extended with strategies of equivalent termination behaviour.
The last column underlines asymptotic performance of the machines.
We can see that weak strategies (call by name, right-to-left and left-to-right call by value) were naturally implemented efficiently, while in the strong case (normal order, twice right-to-left call by value), it is not immediate to avoid an exponential overhead. 

\begin{table}[h]
\begin{center}
\begin{tabular}{ccccc}
machine    & introduced in     & year & strategy as in \cite{I} & overhead    \\
SECD       & \cite{Landin64}   & 1964 & \textit{rcbv}                  & polynomial  \\
Krivine    & \cite{Krivine85}  & 1985 & \textit{cbn}                    & polynomial  \\
CEK        & \cite{DBLP:conf/ifip2/FelleisenF87}     & 1986 & \textit{lcbv}                  & polynomial  \\
KN         & \cite{DBLP:conf/lfp/Cregut90}  & 1990 & \textit{no                    } & exponential \\
Useful MAM & \cite{DBLP:conf/wollic/Accattoli16} & 2016 & \textit{no                    } & polynomial  \\
KNV        & \cite{II}         & 2020 & \textit{rrcbw                 } & exponential \\
SCAM       & \cite{DBLP:conf/lics/AccattoliCC21}       & 2021 &    & polynomial  \\
RKNV       & \cite{III}        & 2021 & \textit{rrcbw}                  & polynomial  \\
RKNL       & \cite{IV}         & 2022 & \textit{no                    } & polynomial 
\end{tabular}
\end{center}
\caption{Selected abstract machines}
\label{tab:machines}
\end{table}

Chapter I surveys reduction strategies of the lambda calculus.
It is first such a~survey since Sestoft's \cite{DBLP:conf/birthday/Sestoft02}.
Moreover,
it explores the space of strategies,
discovers new ones,
and all the gathered strategies are formalized in the Coq proof assistant.
It pays off a part of the research debt (as explained in the previous subchapter) by setting a landscape of the strategies. In particular, it allows us to precisely pinpoint the strategies implemented by the machines presented in Table~\ref{tab:machines}. The paper also compares a few semantic formats and introduces a new one, called \emph{phased semantics}, that facilitates more algebraic reasoning about the strategies and is applied to prove determinism and shapes of the normal forms of the strategies. All of these properties for all the surveyed strategies are formalized in Coq.

Chapter II is the first chronologically published. The titular contribution, namely the KNV machine, is the first published \emph{abstract machine for strong call by value}. It is expressed with notation very close to \cite{DBLP:journals/lisp/Cregut07} because of its close kinship with Crégut's KN from which it was \emph{systematically interderived} via the functional correspondence \cite{DBLP:conf/ppdp/AgerBDM03}. The reader may also witness an old technique of proving progress via annotated decomposition that is replaced with a potential function in the next chapters.
In turn, proofs of the properties of the \textit{rrcbw} strategy could be now imported from \cite{I}. The paper also explains in Section 3.3 a derivation of a bisimilar machine with \emph{explicit shape invariant}. The \emph{invariant derivation} is demonstrated in the code accompanying the paper. Such a bisimilar machine with extra information added for theoretical purposes is named a \emph{ghost machine}\footnote{\emph{post-defence note}: The name ``ghost machine'' mimics ``ghost variables'' from systems such as Why3.} in \cite{IV}. In \cite{II}, it is employed to read back reduction semantics of \textit{rrcbw} by reversing an automaton corresponding to the grammar of machine stacks. Moreover, the paper demonstrates the \emph{term-streaming technique} that can be used in order to short-circuit a conversion check and return early a negative answer, before the given terms are fully normalized.
A drawback of the KNV machine is that there exist lambda terms that can induce an exponential overhead of the execution.
However, the exponential overhead can be stripped out exactly as it was promised in the conclusion section of the chapter.

Chapter III introduces the RKNV machine that reduces the overhead of KNV from exponential to polynomial.
It is one of the two first, besides SCAM, efficient abstract machines for strong call by value, and the one of the two whose notion of strong call by value is a conservative extension of open call by value \cite{DBLP:conf/aplas/AccattoliG16}.
The obtained machine can be intuitively understood thanks to the functionally corresponding normalizer expressed in OCaml.
The guiding intuition eases proving the correctness of the machine.
The proof of the computational complexity requires an amortized analysis.
The analysis is carried out with the use of a \emph{potential function} à la Okasaki.
Applied proof techniques are successfully reused in Chapter IV.
The paper contains also an example showing that term representations with de Bruijn indices or de Bruijn levels without syntax extensions would introduce an unwanted quadratic overhead.

Chapter IV refines the methods employed in \cite{II} and \cite{III} to obtain the first efficient abstract machine for a full-reducing call-by-need strategy, named RKNL.
The machine also improves the state of the art for normal-order reduction by reducing the overhead from quadratic to quasibilinear.
The main part of the derivation via the functional correspondence was done using an automatic tool called \emph{semantic transformer} \cite{BuszkaB21:LOPSTR21}.
Taking into account the complexity of efficient machines for strong call by value and the complexity of call-by-need semantics, 
the obtained machine turns out to be simpler than expected. Its soundness, completeness and efficiency are proven with respect to the normal-order ($\mathit{no}$ of \cite{I}) strategy, which is a strong strategy. Its correctness is proven also w.r.t. a weak call-by-need machine. The paper contains a table summarizing numbers of steps  \emph{measured empirically} revealing the machine's asymptotic performances.
Two ingredients that can be seen as presentational contributions are:
separate presentation of a ghost machine explicitly representing leftmost-outermost contexts,
and an example of a full run of the RKNL machine in \emph{refocusing notation}.

\section{Conclusion}

The thesis of this dissertation is that language is a marvellous tool with plentiful applications.

We have showcased studies on the formal languages enabled by their definitions in metalanguages, formal and informal.
In the case of formal, functional languages (here OCaml and Racket) that are applied in this dissertation, the systematic transformations of definitional interpreters lead to provably asymptotically efficient implementations of strong reductions in the lambda calculus. The reductions, in turn, are applied in partial evaluation, compilation, and proof assistants (see Section \ref{sec:applications} Practical Applications).

Lambda calculus itself can be seen as a very simple, functional, universal programming language that can be used as a metalanguage.
OCaml and Racket implementations can be desugared to the weak call-by-value lambda calculus being a reasonable computation model for both time and space complexity \cite{DBLP:journals/pacmpl/ForsterKR20}. It indicates an interesting self-reinforcement of programming languages research. For example, the Coq proof assistant is used as a tool in Chapter I, while Chapters III and IV potentially have an application in Coq’s implementation. Thus, languages are used to study themselves.

The introduction also presents a synergy between theoretical research and programming.  Chapter I documents a successful attempt at interactive theorem proving. In Chapters II, III, IV, programming enabled fast prototyping of theoretical constructs such as derived abstract machines in OCaml and Racket. On the other hand, tools like Coq, OCaml, and Racket build on the lambda calculus and originate from the research endeavour, let alone \textit{semantic transformer} that is an artefact of new directions in research.

The whole undertaking behind this dissertation was possible thanks to the existence of language. The fact that its whole content is expressed in languages is so obvious that it could remain unnoticed. However, I find it noteworthy to point out that applications of language alone, including reasoning, programming, writing, and communicating, can be recognized as work beneficial to others.

Based on my modest experience, I suppose that the use of lambda calculus as a metalanguage, as a model language, and as an intermediate language still will be bringing fruits and will be showing the effectiveness of language.

\printbibliography

\end{document}